\documentclass[copyright,creativecommons]{eptcs}
\providecommand{\event}{} 

\usepackage{iftex}
\usepackage{amsmath, amssymb, amsthm}
\usepackage{enumitem}
\usepackage{tikz}
\usetikzlibrary{cd} 
\usepackage{graphicx}

\usepackage{tikzit}

\tikzstyle{gate}=[shape=rectangle, text height=1.5ex, text depth=0.25ex, yshift=0.5mm, fill=white, draw=black, minimum height=5mm, yshift=-0.5mm, minimum width=5mm, font={\small}, tikzit category=circuit]
\tikzstyle{big gate}=[shape=rectangle, text height=1.5ex, text depth=0.25ex, yshift=0.5mm, fill=white, draw=black, minimum height=10mm, yshift=-0.5mm, minimum width=5mm, font={\small}, tikzit category=circuit]
\tikzstyle{Z dot}=[inner sep=0mm, minimum size=2mm, shape=circle, draw=black, fill={rgb,255: red,221; green,255; blue,221}, tikzit category=zx]
\tikzstyle{Z phase dot}=[minimum size=5mm, font={\footnotesize\boldmath}, shape=rectangle, rounded corners=2mm, inner sep=0.2mm, outer sep=-2mm, scale=0.8, tikzit shape=circle, draw=black, fill={rgb,255: red,221; green,255; blue,221}, tikzit draw=blue, tikzit category=zx]
\tikzstyle{X dot}=[Z dot, shape=circle, draw=black, fill={rgb,255: red,255; green,136; blue,136}, tikzit category=zx]
\tikzstyle{X phase dot}=[Z phase dot, tikzit shape=circle, tikzit draw=blue, fill={rgb,255: red,255; green,136; blue,136}, font={\footnotesize\boldmath}, tikzit category=zx]
\tikzstyle{hadamard}=[fill=yellow, draw=black, shape=rectangle, inner sep=0.6mm, minimum height=1.5mm, minimum width=1.5mm, tikzit category=zx]
\tikzstyle{paulibox}=[fill={rgb,255: red,221; green,221; blue,255}, draw=black, shape=rectangle, inner sep=0.6mm, minimum height=5mm, minimum width=5mm, font={\footnotesize}, text height=1.5ex, text depth=0.25ex, tikzit category=zx]
\tikzstyle{vertex}=[inner sep=0mm, minimum size=1mm, shape=circle, draw=black, fill=black, tikzit category=misc]
\tikzstyle{vertex set}=[inner sep=0mm, minimum size=1mm, shape=circle, draw=black, fill=white, font={\footnotesize\boldmath}, tikzit category=misc]
\tikzstyle{small black dot}=[fill=black, draw=black, shape=circle, inner sep=0pt, minimum width=1.2mm, tikzit category=circuit]
\tikzstyle{cnot ctrl}=[fill=black, draw=black, shape=circle, inner sep=0pt, minimum width=1.2mm, tikzit category=circuit]
\tikzstyle{cnot targ}=[fill=white, draw=white, shape=circle, tikzit category=circuit, label={center:$\oplus$}, inner sep=0pt, minimum width=2.1mm, tikzit fill={rgb,255: red,102; green,204; blue,255}, tikzit draw=black]
\tikzstyle{ket}=[fill=white, draw=black, shape=regular polygon, regular polygon sides=3, regular polygon rotate=-30, scale=0.7, inner sep=1pt, tikzit category=circuit, tikzit shape=rectangle, tikzit fill=green]
\tikzstyle{bra}=[fill=white, draw=black, shape=regular polygon, regular polygon sides=3, regular polygon rotate=30, scale=0.7, inner sep=1pt, tikzit category=circuit, tikzit shape=rectangle, tikzit fill=red]
\tikzstyle{scalar}=[shape=rectangle, text height=1.5ex, text depth=0.25ex, yshift=0.5mm, fill=white, draw=black, minimum height=5mm, yshift=-0.5mm, minimum width=5mm, font={\small}]
\tikzstyle{clabel}=[fill=white, draw=none, shape=rectangle, tikzit fill={rgb,255: red,56; green,255; blue,242}, font={\footnotesize}, inner sep=1pt, tikzit category=labels]
\tikzstyle{empty diagram}=[draw={gray!40!white}, dashed, shape=rectangle, minimum width=1cm, minimum height=1cm, tikzit category=misc]
\tikzstyle{amap}=[fill=white, draw=black, shape=NEbox, tikzit category=asymmetric, tikzit fill=yellow, tikzit shape=rectangle]
\tikzstyle{amap conj}=[fill=white, draw=black, shape=NWbox, tikzit category=asymmetric, tikzit fill=green, tikzit shape=rectangle]
\tikzstyle{amap adj}=[fill=white, draw=black, shape=SEbox, tikzit category=asymmetric, tikzit fill=red, tikzit shape=rectangle]
\tikzstyle{amap trans}=[fill=white, draw=black, shape=SWbox, tikzit category=asymmetric, tikzit fill=orange, tikzit shape=rectangle]
\tikzstyle{astate}=[fill=white, draw=black, shape=NEtriangle, tikzit category=asymmetric, tikzit shape=circle, tikzit fill=yellow]
\tikzstyle{astate conj}=[fill=white, draw=black, shape=NWtriangle, tikzit category=asymmetric, tikzit shape=circle, tikzit fill=green]
\tikzstyle{astate adj}=[fill=white, draw=black, shape=SEtriangle, tikzit category=asymmetric, tikzit shape=circle, tikzit fill=red]
\tikzstyle{astate trans}=[fill=white, draw=black, shape=SWtriangle, tikzit category=asymmetric, tikzit shape=circle, tikzit fill=orange]
\tikzstyle{Doubled X Phase dot}=[fill={rgb,255: red,255; green,136; blue,136}, draw=black, shape=rectangle, tikzit category=zx, fill={rgb,255: red,255; green,136; blue,136}, font={\footnotesize\boldmath}, tikzit draw=blue, tikzit shape=circle, line width=1.6pt, rounded corners=2mm, inner sep=1mm, outer sep=-2mm, minimum size=4mm]
\tikzstyle{doubled z phase dot}=[draw=black, shape=rectangle, tikzit category=zx, fill={rgb,255: red,221; green,255; blue,221}, font={\footnotesize\boldmath}, tikzit draw=blue, tikzit shape=circle, line width=1.6pt, rounded corners=2mm, inner sep=1mm, outer sep=-2mm, minimum size=4mm]
\tikzstyle{doubled hadamard}=[fill=yellow, draw=black, shape=rectangle, inner sep=0.6mm, minimum height=1.5mm, minimum width=1.5mm, line width=1.6pt, minimum size=3mm, tikzit category=zx]
\tikzstyle{medium box}=[fill=white, draw=black, shape=rectangle, minimum height=1cm, minimum width=0.75cm]

\tikzstyle{hadamard edge}=[-, dashed, dash pattern=on 2pt off 0.5pt, thick, draw={rgb,255: red,68; green,136; blue,255}]
\tikzstyle{box edge}=[-, dashed, dash pattern=on 2pt off 0.5pt, thick, draw={rgb,255: red,203; green,192; blue,225}]
\tikzstyle{brace edge}=[-, tikzit draw=blue, decorate, decoration={brace,amplitude=1mm,raise=-1mm}]
\tikzstyle{diredge}=[->]
\tikzstyle{double edge}=[-, double, shorten <=-1mm, shorten >=-1mm, double distance=2pt]
\tikzstyle{gray edge}=[-, {gray!60!white}]
\tikzstyle{pointer edge}=[->, very thick, gray]
\tikzstyle{boldedge}=[-, line width=1.6pt, shorten <=-0.17mm, shorten >=-0.17mm]
\tikzstyle{bidir edge}=[<->, very thick, draw={rgb,255: red,191; green,191; blue,191}]
\tikzstyle{functorbox}=[-, fill=none, draw={rgb,255: red,150; green,150; blue,189}, line width=1.2pt, dashed, dash pattern=on 5pt off 2pt]

\ifpdf
  \usepackage{underscore}         
  \usepackage[T1]{fontenc}        
\else
  \usepackage{breakurl}           
\fi
\usepackage[only,llbracket,rrbracket]{stmaryrd}

\newcommand{\Smooth}{\mathsf{Smooth}}
\newcommand{\FHilb}{\mathsf{FHilb_{\mathbb{C}}}}
\newcommand{\Qubit}{\mathsf{Qubit}}

\newcommand{\CPM}{\mathsf{CPM(\FHilb)}}
\newcommand{\CPMQ}{\mathsf{CPM(\Qubit)}}
\newcommand{\doublezx}{\mathsf{ZX_{\scalebox{0.25}{\begin{tikzpicture}
	\begin{pgfonlayer}{nodelayer}
		\node [style=none] (1) at (0, -0.75) {};
		\node [style=none] (2) at (0, 0.25) {};
		\node [style=none] (3) at (1.05, 0.25) {};
		\node [style=none] (4) at (-1.05, 0.25) {};
		\node [style=none] (5) at (0.65, 0.65) {};
		\node [style=none] (6) at (-0.65, 0.65) {};
		\node [style=none] (7) at (0.25, 1.05) {};
		\node [style=none] (8) at (-0.25, 1.05) {};
	\end{pgfonlayer}
	\begin{pgfonlayer}{edgelayer}
		\draw [style=boldedge] (3.center) to (4.center);
		\draw [style=boldedge] (5.center) to (6.center);
		\draw [style=boldedge] (7.center) to (8.center);
		\draw [style=boldedge] (2.center) to (1.center);
	\end{pgfonlayer}
\end{tikzpicture}}}}}
\newcommand{\MatR}{\mathsf{Mat}_\mathbb{R}}

\newcommand{\boxfunctor}{\mathsf{F}}
\newcommand{\embedfunctor}{\mathsf{G}}
\newcommand{\interpF}[1][-]{\mathsf{\llbracket #1 \rrbracket}}
\newcommand{\interpFbig}[1][-]{\mathsf{\Bigl\llbracket \ #1 \ \Bigr\rrbracket}}
\newcommand{\interpFbigg}[1][-]{\mathsf{\Biggl\llbracket \ #1 \ \Biggr\rrbracket}}
\newcommand{\smoothfunctor}{\mathsf{H}}

\newcommand{\bra}[1]{\lvert #1 \rangle}
\newcommand{\ket}[1]{\langle #1 \rvert}
\newcommand{\expectation}[1]{\langle #1 \rangle}

\def\fcmp{\mathbin{\raise 0.6ex\hbox{\oalign{\hfil$\scriptscriptstyle      \mathrm{o}$\hfil\cr\hfil$\scriptscriptstyle\mathrm{9}$\hfil}}}}

\newtheorem{theorem}{Theorem}
\newtheorem{corollary}{Corollary}

\theoremstyle{remark}
\newtheorem{remark}{Remark}
\theoremstyle{definition}
\newtheorem{definition}{Definition}

\title{Hybrid Quantum-Classical Machine Learning with String Diagrams}
\author{Alexander Koziell-Pipe
\institute{University of Oxford\\
Oxford, UK}
\email{alexander.koziell-pipe@cs.ox.ac.uk}
\and
Aleks Kissinger
\institute{University of Oxford\\
Oxford, UK}
\email{aleks.kissinger@cs.ox.ac.uk}
}
\def\titlerunning{Diagrammatic Quantum Machine Learning}
\def\authorrunning{A. Koziell-Pipe \& A. Kissinger}
\begin{document}
\maketitle

\begin{abstract}
Central to near-term quantum machine learning is the use of hybrid quantum-classical algorithms. This paper develops a formal framework for describing these algorithms in terms of string diagrams: a key step towards integrating these hybrid algorithms into existing work using string diagrams for machine learning and differentiable programming. A notable feature of our string diagrams is the use of functor boxes, which correspond to a quantum-classical interfaces. The functor used is a lax monoidal functor embedding the quantum systems into classical, and the lax monoidality imposes restrictions on the string diagrams when extracting classical data from quantum systems via measurement. In this way, our framework provides initial steps toward a denotational semantics for hybrid quantum machine learning algorithms that captures important features of quantum-classical interactions.
\end{abstract}

\section{Introduction}

At time of writing, we are still in the Noisy Intermediate Scale (NISQ) era of quantum computing\cite{Preskill2018}, where techniques to correct errors are still under development and the size of quantum circuits are limited. Quantum algorithms requiring deep circuits and larger numbers of qubits, notably Shor's\cite{Shor1994, Shor1997} and Grover's\cite{Grover1996} algorithms, are still beyond our reach for practical use.

Hybrid quantum-classical algorithms, where quantum computations interoperate with classical, 
offer a pragmatic solution to these limitations. These hybrid algorithms extract as much performance as possible out of limited quantum resources by delegating a large part of their computation to classical computers, then using quantum processors for small subroutines for which they are particularly well-suited.

The archetypal example of hybrid algorithms in the context of quantum machine learning is the Variational Quantum Eigensolver\cite{Peruzzo2014} and more generally, Variational Quantum Algorithms\cite{Cerezo2021}. In these algorithms a parameterized quantum circuit is used to compute a cost function, which is provided to a classical optimizer that computes updated parameters to minimize the cost. Computing updated parameters may involve further quantum processing in order to evaluate gradients of the cost function with respect to the parameters\cite{Mitarai2018, Schuld2019}. This process is repeated iteratively until a desired performance or termination condition is met. These algorithms are inherently hybrid in the sense of \cite{Callison2022}, as they require non-trivial amounts of both quantum and classical computational resources to run. Furthermore, they involve the repeated transfer of information between quantum and classical systems.

In this work, we provide a diagrammatic formalism of hybrid algorithms using category theory. Originally developed to study algebraic topology\cite{eilenberg1945general}, category theory has since found applications across a range of scientific disciplines, including quantum computing\cite{abramsky2004categorical, abramsky2008categorical, coecke2018picturing} and machine learning\cite{fong2019backprop, cruttwell2024deep, gavranović2024position}\footnote{For references spanning a range of machine learning topics, see \cite{githubCTML}.}.

Category theory provides a unifying language for different branches of mathematics. Furthermore, it offers an abstract framework for understanding and formalizing different mathematical structures and the relationships between them. In the context of this paper, category theory bridges the gap between machine learning and quantum computing. The mathematics of machine learning is characterised by smooth, differentiable functions and the ability to freely copy and delete information without restriction. This sits in stark contrast with quantum computation, where calculating gradients of arbitrary quantum circuits is highly non-trivial, and copying and deleting information is physically prohibited by the no-cloning and no-deleting theorems of quantum information.

This separation between the mathematics of machine learning and quantum computing is reflected by the distinct properties of the categories in which each are described. While categorical quantum mechanics is formulated in terms of dagger compact closed categories, the categorical machine learning typically uses Cartesian categories in which there is a suitable notion of a reverse derivative. Our work is a step towards combining these two theories of systems and processes into a single framework, with the intention of facilitating seamless interoperability in practical implementations of hybrid quantum-classical algorithms.

\begin{figure}[t]
    \centering
    \includegraphics[width=12cm]{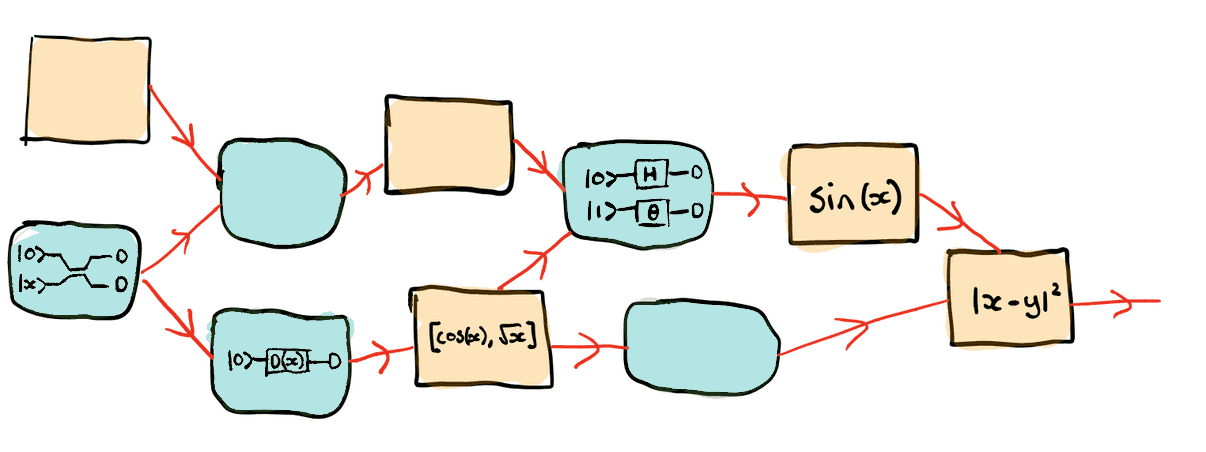}
    \caption{From \cite{pennylaneHybridComputation}. An informal schematic of the flow of information between classical and quantum nodes. Our work aims to formalize this in terms of category theory, allowing it to be depicted as a string diagram.}
\end{figure}

Our contributions are as follows. First, following\cite{cruttwell2024deep}, we establish the Cartesian category $\Smooth$ as the category with which to describe the classical part of our algorithms. Next, we choose the compact closed category $\CPM$ of completely positive maps\cite{selinger2007dagger} for the quantum part. We then show that there exists a lax monoidal functor $\boxfunctor: \CPM \longrightarrow \Smooth$, which we define as the composition of two functors $\embedfunctor: \CPM \longrightarrow
\MatR$ and $\smoothfunctor: \MatR \longrightarrow \Smooth$. We give a concrete description of this functor by fixing a basis in terms of Pauli operators for each object of $\CPM$. The lax monoidal functor bridges our quantum and classical systems, and allows the description of hybrid classical-quantum algorithms using string diagrams\cite{selinger2010survey, piedeleu2023introduction} in which functorial boxes\cite{cockett1999linearly, mellies2006functorial} represent classical-quantum interfaces. The lax condition on our functor allows multiple wires to enter the box, representing the encoding of classical data into the quantum circuit, but only one wire to exit, representing a measurement probability. Furthermore, we are able to represent the quantum system as a ZX-diagram\cite{coecke2008interacting, vandewetering2020zxcalculus, websitezxcalculus}, making the quantum part of our algorithm amenable to quantum circuit optimization methods. Finally, we show a concrete application of our framework to the quantum machine learning algorithm of \cite{farhi2018classification}, in which parameterized quantum circuits are trained on MNIST classification\cite{MNIST}. Our framework gives us a denotational semantics of this quantum machine learning algorithm that could form the basis of a programming language for quantum machine learning, or hybrid quantum-classical algorithms more generally.

Our work is intended as a step toward the extension of categorical gradient based learning framework to allow `quantum nodes', and future work will aim to formalise parameterization and backpropagation in these hybrid algorithms through the use of the parametric lens construction of \cite{cruttwell2024deep} and diagrams for `gradient recipes'\cite{Mitarai2018, Schuld2019}, an established method for computing the gradient of a parameterized quantum circuit for use in quantum machine learning.
\section{Theory}

In this section, we define categories $\CPMQ$ and $\Smooth$ for quantum and classical systems. We then show that there is a lax monoidal functor $\boxfunctor: \CPMQ \rightarrow \Smooth$. We define as the composition of two functors  $\embedfunctor: \CPMQ \rightarrow \MatR$ and $\smoothfunctor: \MatR \rightarrow \Smooth$, using the category $\MatR$ of real matrices as an intermediate step in the passage from quantum to classical systems.

For a brief review of categories and functors, see section \ref{appendixCT} of the appendix. For a further introduction to category theory see, for example, \cite{leinster2016basic, fong2018seven}.

\subsection{$\FHilb$ and ZX-diagrams}

Pure state quantum mechanics can be described mathematically in terms of complex Hilbert spaces. Furthermore, in quantum computing we are usually interested in finite-dimensional systems. As a consquence, much of the work applying category theory to quantum mechanics involves the use of the category $\FHilb$:

\begin{definition}[$\FHilb$]
    The symmetric monoidal category $\FHilb$ is defined by:
    \begin{description}[labelindent=2.5em]
        \item[Objects] finite-dimensional complex Hilbert spaces.
        \item[Morphisms] $H \rightarrow K$ linear maps from $H$ to $K$.
        \item[Monoidal product] the tensor product of Hilbert spaces.
        \item[Monoidal unit] the 0-dimensional Hilbert space.
    \end{description}
\end{definition}

When working with quantum computation, in particular qubits, we can simplify the objects we have to work with by moving to the subcategory of $\FHilb$ consisting of products of 2-dimensional complex spaces $(\mathbb{C}^{2})^{\otimes n}$. We call this category $\Qubit$ and define it as a PROP:

\begin{definition}[$\Qubit$]
    The category $\Qubit$ is the PROP whose morphisms $n \rightarrow m$ are given by $2^m \times 2^n$ complex matrices, $\Qubit(n, m) := \mathbb{M}^{2^m \times 2^n}_{\mathbb{C}}$. The monoidal product on morphisms is given by the Kronecker product of matrices.
\end{definition}

We can think of the objects $n \in \Qubit$ as corresponding to Hilbert spaces $(\mathbb{C}^2)^{\otimes n}$ containing the complex amplitudes for an $n$-qubit system. When working in $\Qubit$, we can use ZX-diagrams\cite{coecke2008interacting, vandewetering2020zxcalculus, websitezxcalculus} to describe our morphisms. ZX-diagrams are generated by green (light) Z spiders and red (dark) X spiders, corresponding to the quantum states $\bra{00...0}\ket{00..0} + e^{i\alpha}\bra{11...1}\ket{11...1}$ and $\bra{++...+}\ket{++...+} + e^{i\alpha}\bra{--...-}\ket{--...-}$, respectively:

\begin{equation*}
    \begin{tikzpicture}
	\begin{pgfonlayer}{nodelayer}
		\node [style=none] (1) at (2.5, 2) {};
		\node [style=none] (3) at (2.5, -2) {};
		\node [style=none] (4) at (-2.5, 2) {};
		\node [style=Z phase dot] (5) at (0, 0) {$\alpha$};
		\node [style=none] (6) at (-2.5, -2) {};
		\node [style=none] (7) at (-2.25, -0.5) {$\vdots$};
		\node [style=none] (8) at (2.25, -0.5) {$\vdots$};
		\node [style=none] (9) at (-2.5, 1) {};
		\node [style=none] (10) at (2.5, 0.75) {};
	\end{pgfonlayer}
	\begin{pgfonlayer}{edgelayer}
		\draw [in=0, out=-135] (5) to (6.center);
		\draw [in=0, out=135] (5) to (4.center);
		\draw [in=180, out=-45] (5) to (3.center);
		\draw [in=-180, out=45] (5) to (1.center);
		\draw [in=0, out=165] (5) to (9.center);
		\draw [in=0, out=-180] (10.center) to (5);
	\end{pgfonlayer}
\end{tikzpicture} \qquad \qquad \qquad \begin{tikzpicture}
	\begin{pgfonlayer}{nodelayer}
		\node [style=none] (1) at (2.5, 2) {};
		\node [style=none] (3) at (2.5, -2) {};
		\node [style=none] (4) at (-2.5, 2) {};
		\node [style=X phase dot] (5) at (0, 0) {$\alpha$};
		\node [style=none] (6) at (-2.5, -2) {};
		\node [style=none] (7) at (-2.25, -0.5) {$\vdots$};
		\node [style=none] (8) at (2.25, -0.5) {$\vdots$};
		\node [style=none] (9) at (-2.5, 1) {};
		\node [style=none] (10) at (2.5, 0.75) {};
	\end{pgfonlayer}
	\begin{pgfonlayer}{edgelayer}
		\draw [in=0, out=-135] (5) to (6.center);
		\draw [in=0, out=135] (5) to (4.center);
		\draw [in=180, out=-45] (5) to (3.center);
		\draw [in=-180, out=45] (5) to (1.center);
		\draw [in=0, out=165] (5) to (9.center);
		\draw [in=0, out=-180] (10.center) to (5);
	\end{pgfonlayer}
\end{tikzpicture}
\end{equation*}

Any morphism in $\Qubit$ can be expressed as (possibly multiple) ZX-diagram(s). Furthermore, the diagrams can be equipped with a set of rewrite rules allowing one diagram to be transformed into another. There are known sets of rewrite rules that are sound and complete, meaning that two diagrams can be transformed into one another if and only if they correspond to the same quantum process. ZX-diagrams together with a canonical set of rewrite rules are referred to as the ZX-calculus, and has useful applications to quantum circuit simplification and quantum error correction. For a more in-depth introduction, see \cite{vandewetering2020zxcalculus}.

\subsection{$\CPM$ and Doubled ZX}

When bridging from quantum to classical systems, we will choose to work with the density operator formalism of quantum mechanics. This allows us to work with mixed-state quantum mechanics, as well as express our diagrams using a basis of Pauli operators with real coefficients. This requires deriving a different category from $\FHilb$, using the CPM-construction of \cite{selinger2007dagger}. When applied to $\FHilb$, this gives us the following category, where we denote the space of linear operators on a Hilbert space $H$ by $\mathcal{L}(H)$:

\begin{definition}
    The category $\CPM$ is defined by
    \begin{description}[labelindent=2.5em]
        \item[Objects] complex Hilbert spaces.
        \item[Morphisms] $H \rightarrow K$ completely positive linear superoperators $\mathcal{L}(H) \longrightarrow \mathcal{L}(K)$.
    \end{description}
\end{definition}

Like $\FHilb$, $\CPM$ is symmetric monoidal, with monoidal product given by the tensor product. If we once again restrict to 2-dimensional complex Hilbert spaces, we can similarly define $\CPMQ$:

\begin{definition}
    The category $\CPMQ$ is the PROP whose morphisms $n \rightarrow m$ are the completely positive linear superoperators $\mathcal{L}((\mathbb{C}^2)^{\otimes n}) \longrightarrow \mathcal{L}((\mathbb{C}^2)^{\otimes m})$.
\end{definition}

Contrasting with $\Qubit$, objects $n$ of $\CPMQ$ can be thought of as corresponding to spaces $(\mathbb{C}^2)^{\otimes n} \otimes (\mathbb{C}^2)^{\otimes n} := \ \mathbb{C}^{2^n \times 2^n}$. States (morphisms $0 \rightarrow n$ for some $n \in \Qubit$) correspond to density matrices, while processes $n \rightarrow m$ correspond to quantum channels.

We can extend the ZX-calculus to $\CPMQ$ through a process known as `doubling'\cite[Chapter 8]{coecke2018picturing}. In this doubled notation, thick (`doubled') wires correspond to quantum systems and thin single wires correspond to classical ones. Z and X spiders in this doubled notation are defined in terms of their undoubled counterparts as follows:

\begin{equation}
    \begin{tikzpicture}
	\begin{pgfonlayer}{nodelayer}
		\node [style=none] (1) at (2.5, 2) {};
		\node [style=none] (3) at (2.5, -2) {};
		\node [style=none] (4) at (-2.5, 2) {};
		\node [style=doubled z phase dot] (5) at (0, 0) {$\alpha$};
		\node [style=none] (6) at (-2.5, -2) {};
		\node [style=none] (7) at (-2.25, -0.5) {$\vdots$};
		\node [style=none] (8) at (2.25, -0.5) {$\vdots$};
		\node [style=none] (9) at (-2.5, 1) {};
		\node [style=none] (10) at (2.5, 0.75) {};
	\end{pgfonlayer}
	\begin{pgfonlayer}{edgelayer}
		\draw [style=boldedge, in=0, out=-135] (5) to (6.center);
		\draw [style=boldedge, in=0, out=135] (5) to (4.center);
		\draw [style=boldedge, in=180, out=-45] (5) to (3.center);
		\draw [style=boldedge, in=-180, out=45] (5) to (1.center);
		\draw [style=boldedge, in=0, out=165] (5) to (9.center);
		\draw [style=boldedge, in=0, out=-180] (10.center) to (5);
	\end{pgfonlayer}
\end{tikzpicture} \quad := \quad \begin{tikzpicture}
	\begin{pgfonlayer}{nodelayer}
		\node [style=none] (1) at (2.5, 2.375) {};
		\node [style=none] (3) at (2.5, -1.875) {};
		\node [style=none] (4) at (-2.5, 2.375) {};
		\node [style=Z phase dot] (5) at (0, 0.725) {$\alpha$};
		\node [style=none] (6) at (-2.5, -1.875) {};
		\node [style=none] (9) at (-2.5, 0.775) {};
		\node [style=none] (10) at (2.5, 0.725) {};
		\node [style=none] (11) at (2.5, 1.875) {};
		\node [style=none] (12) at (2.5, -2.375) {};
		\node [style=none] (13) at (-2.5, 1.875) {};
		\node [style=Z phase dot] (14) at (0, -0.775) {$-\alpha$};
		\node [style=none] (15) at (-2.5, -2.375) {};
		\node [style=none] (16) at (-2.25, -0.875) {$\vdots$};
		\node [style=none] (17) at (2.25, -0.875) {$\vdots$};
		\node [style=none] (18) at (-2.5, 0.375) {};
		\node [style=none] (19) at (2.5, 0.325) {};
	\end{pgfonlayer}
	\begin{pgfonlayer}{edgelayer}
		\draw [in=0, out=-135] (5) to (6.center);
		\draw [in=0, out=135] (5) to (4.center);
		\draw [in=180, out=-45] (5) to (3.center);
		\draw [in=-180, out=45] (5) to (1.center);
		\draw [in=0, out=165] (5) to (9.center);
		\draw [in=15, out=-180] (10.center) to (5);
		\draw [in=0, out=-135] (14) to (15.center);
		\draw [in=0, out=135] (14) to (13.center);
		\draw [in=180, out=-45] (14) to (12.center);
		\draw [in=-180, out=45] (14) to (11.center);
		\draw [in=0, out=165] (14) to (18.center);
		\draw [in=0, out=-180] (19.center) to (14);
	\end{pgfonlayer}
\end{tikzpicture} \qquad\qquad
    \begin{tikzpicture}
	\begin{pgfonlayer}{nodelayer}
		\node [style=none] (1) at (2.5, 2) {};
		\node [style=none] (3) at (2.5, -2) {};
		\node [style=none] (4) at (-2.5, 2) {};
		\node [style=Doubled X Phase dot] (5) at (0, 0) {$\alpha$};
		\node [style=none] (6) at (-2.5, -2) {};
		\node [style=none] (7) at (-2.25, -0.5) {$\vdots$};
		\node [style=none] (8) at (2.25, -0.5) {$\vdots$};
		\node [style=none] (9) at (-2.5, 1) {};
		\node [style=none] (10) at (2.5, 0.75) {};
	\end{pgfonlayer}
	\begin{pgfonlayer}{edgelayer}
		\draw [style=boldedge, in=0, out=-135] (5) to (6.center);
		\draw [style=boldedge, in=0, out=135] (5) to (4.center);
		\draw [style=boldedge, in=180, out=-45] (5) to (3.center);
		\draw [style=boldedge, in=-180, out=45] (5) to (1.center);
		\draw [style=boldedge, in=0, out=165] (5) to (9.center);
		\draw [style=boldedge, in=0, out=-180] (10.center) to (5);
	\end{pgfonlayer}
\end{tikzpicture} \quad := \quad \begin{tikzpicture}
	\begin{pgfonlayer}{nodelayer}
		\node [style=none] (1) at (2.5, 2.375) {};
		\node [style=none] (3) at (2.5, -1.875) {};
		\node [style=none] (4) at (-2.5, 2.375) {};
		\node [style=X phase dot] (5) at (0, 0.725) {$\alpha$};
		\node [style=none] (6) at (-2.5, -1.875) {};
		\node [style=none] (9) at (-2.5, 0.775) {};
		\node [style=none] (10) at (2.5, 0.725) {};
		\node [style=none] (11) at (2.5, 1.875) {};
		\node [style=none] (12) at (2.5, -2.375) {};
		\node [style=none] (13) at (-2.5, 1.875) {};
		\node [style=X phase dot] (14) at (0, -0.775) {$-\alpha$};
		\node [style=none] (15) at (-2.5, -2.375) {};
		\node [style=none] (16) at (-2.25, -0.875) {$\vdots$};
		\node [style=none] (17) at (2.25, -0.875) {$\vdots$};
		\node [style=none] (18) at (-2.5, 0.375) {};
		\node [style=none] (19) at (2.5, 0.325) {};
	\end{pgfonlayer}
	\begin{pgfonlayer}{edgelayer}
		\draw [in=0, out=-135] (5) to (6.center);
		\draw [in=0, out=135] (5) to (4.center);
		\draw [in=180, out=-45] (5) to (3.center);
		\draw [in=-180, out=45] (5) to (1.center);
		\draw [in=0, out=165] (5) to (9.center);
		\draw [in=15, out=-180] (10.center) to (5);
		\draw [in=0, out=-135] (14) to (15.center);
		\draw [in=0, out=135] (14) to (13.center);
		\draw [in=180, out=-45] (14) to (12.center);
		\draw [in=-180, out=45] (14) to (11.center);
		\draw [in=0, out=165] (14) to (18.center);
		\draw [in=0, out=-180] (19.center) to (14);
	\end{pgfonlayer}
\end{tikzpicture}
\end{equation}

In addition, we gain a `decoherence' operation, defined by:

\begin{equation}
    \begin{tikzpicture}
	\begin{pgfonlayer}{nodelayer}
		\node [style=none] (1) at (-1.5, 0) {};
		\node [style=none] (2) at (0.5, 0) {};
		\node [style=none] (3) at (0.5, 1.25) {};
		\node [style=none] (4) at (0.5, -1.25) {};
		\node [style=none] (5) at (1, 0.75) {};
		\node [style=none] (6) at (1, -0.75) {};
		\node [style=none] (7) at (1.5, 0.25) {};
		\node [style=none] (8) at (1.5, -0.25) {};
	\end{pgfonlayer}
	\begin{pgfonlayer}{edgelayer}
		\draw [style=boldedge] (3.center) to (4.center);
		\draw [style=boldedge] (5.center) to (6.center);
		\draw [style=boldedge] (7.center) to (8.center);
		\draw [style=boldedge] (2.center) to (1.center);
	\end{pgfonlayer}
\end{tikzpicture} \quad := \quad \begin{tikzpicture}
	\begin{pgfonlayer}{nodelayer}
		\node [style=none] (0) at (0.75, 1) {};
		\node [style=none] (1) at (0.75, -1) {};
		\node [style=none] (2) at (0.5, 1) {};
		\node [style=none] (3) at (0.5, -1) {};
		\node [style=none] (4) at (-2, 0.75) {};
		\node [style=none] (5) at (-2, -0.75) {};
	\end{pgfonlayer}
	\begin{pgfonlayer}{edgelayer}
		\draw [bend left=90, looseness=2.00] (0.center) to (1.center);
		\draw (2.center) to (0.center);
		\draw [in=180, out=0] (3.center) to (1.center);
		\draw [in=0, out=-180] (2.center) to (4.center);
		\draw [in=180, out=0] (5.center) to (3.center);
	\end{pgfonlayer}
\end{tikzpicture}
\end{equation}

In terms of density operators and superoperators, we can regard this decoherence operator as taking the trace, or partial trace if applied to a single output of a multiple-output diagram. In the context of quantum circuits, this operation can be thought of as discarding the measurement result on a qubit, and is often referred to as the `discard' operation.

These doubled ZX diagrams live in their own PROP, which we denote $\doublezx$. There is then an interpretation functor $\interpF: \doublezx \rightarrow \CPMQ$ that maps a doubled ZX diagram to the completely positive map it corresponds to.

\subsection{$\Smooth$}

We describe the classical part of our machine learning algorithms using the category $\Smooth$ of smooth maps between Euclidean spaces. We choose to define $\Smooth$ as a PROP, with objects $n$ corresponding to Euclidean spaces $\mathbb{R}^n$:

\begin{definition}
    The category $\Smooth$ is the PROP with morphisms $n \rightarrow m$ tuples of smooth maps $f := (f_1, \ldots, f_m): \mathbb{R}^n \longrightarrow \mathbb{R}^m$. The monoidal product of two morphisms $f: n \rightarrow m$ and $g: n' \rightarrow m'$ is obtained by concatenating the output of $f$ (an $m$-tuple) on the first $n$ elements of $\mathbb{R}^{n + n'}$ with the output of $g$ (an $m'$-tuple) on the last $n'$ elements.
\end{definition}

\subsection{Bridging from Quantum to Classical}

In order to combine quantum and classical systems in the same string diagrams, we must define a functor $\boxfunctor: \CPMQ \rightarrow \Smooth$. To this end, we first introduce the PROP $\MatR$:

\begin{definition}
    The category $\MatR$ is the PROP whose morphisms $n \rightarrow m$ are the real $m \times n$ matrices $\mathbb{M}^{m \times n}_{\mathbb{R}}$. The monoidal product on morphisms is given by the Kronecker product of matrices.
\end{definition}

Next, we define a functor $\embedfunctor: \CPMQ \rightarrow \MatR$, which allows us to express quantum computations using real matrices. First, note that a completely positive superoperator $\mathcal{L}(\mathbb{C}^{2^n}) \longrightarrow \mathcal{L}(\mathbb{C}^{2^m})$ is uniquely determined by a real matrix with respect to the $n$ and $m$-tensor products of Pauli matrices.

 \begin{definition}
     The functor $\embedfunctor: \CPMQ \longrightarrow \MatR$ is defined by the following mappings:
     \begin{description}[labelindent=2.5em]
        \item[Objects] $n \longmapsto 4^n$. 
        \item[Morphisms] A completely positive map $\mathcal{L}(\mathbb{C}^{2^n}) \longrightarrow \mathcal{L}(\mathbb{C}^{2^m})$ is mapped to the $m \times n$ real matrix with respect to the $m$- and $n$-tensor Pauli bases.
     \end{description}
 \end{definition}

As an example, consider the quantum state $\bra{0} + e^{i\alpha}\bra{1}$. Ignoring global scalar factors and working in the computational basis, the density operator for this state is:
\begin{equation*}
    \begin{pmatrix}
        1 & e^{-i\alpha} \\
        e^{i\alpha} & 1
    \end{pmatrix} \quad = \quad I + \cos (\alpha) \sigma_x + \sin (\alpha) \sigma_y
\end{equation*}

Where $I$ is the identity matrix and $\sigma_i$ are the Pauli matrices. This would correspond to the following doubled ZX diagram:

\begin{equation*}
    \begin{tikzpicture}
	\begin{pgfonlayer}{nodelayer}
		\node [style=doubled z phase dot] (0) at (-0.75, 0) {$\alpha$};
		\node [style=none] (1) at (1, 0) {};
	\end{pgfonlayer}
	\begin{pgfonlayer}{edgelayer}
		\draw [style=boldedge] (0) to (1.center);
	\end{pgfonlayer}
\end{tikzpicture}

    \vspace{0.5em}
\end{equation*}

Interpreting this diagram in $\CPMQ$ and applying $\embedfunctor$, this would become:

\begin{equation}
    \embedfunctor\interpFbig[\ \begin{tikzpicture}
	\begin{pgfonlayer}{nodelayer}
		\node [style=doubled z phase dot] (0) at (-0.75, 0) {$\alpha$};
		\node [style=none] (1) at (1, 0) {};
	\end{pgfonlayer}
	\begin{pgfonlayer}{edgelayer}
		\draw [style=boldedge] (0) to (1.center);
	\end{pgfonlayer}
\end{tikzpicture}
 \ ] \quad = \quad \begin{pmatrix}
        1 \\ \cos(\alpha) \\ \sin (\alpha) \\ 0
    \end{pmatrix}
\end{equation}

More generally, we can derive real matrices for completely positive maps as follows.  For each $n \in \mathbb{N}$, fix an ordering on the $n$-qubit Pauli operators \footnote{For example, $\left( I \otimes I,\ I \otimes \sigma_{X},\ I \otimes \sigma_{Y},\ I \otimes \sigma_{Z},\ \sigma_{X} \otimes I,\ ...\ ,\ \sigma_{Z} \otimes \sigma_{Z}\right)$ on the 2-qubit operators.} and label them according to this order as $\sigma^{n}_i$ for $i \in \{ 0,\ 1,\ ...\ ,\ 4^n-1 \}$. Then entries of the matrix $M$ of a completely positive superoperator $\Phi: n \rightarrow m$ are obtained via:

\begin{equation}\label{eq:cptomat}
    M_{ij} = \mathrm{tr}\left(\left(\sigma^{m}_i\right)^\dagger \cdot\ \Phi(\sigma^{n}_j)\right)
\end{equation}

Given a Kraus decomposition of the operator $\Phi$, this can be calculated numerically. For example,  we can use (\ref{eq:cptomat}) to derive real matrices for the state $\bra{+} + e^{i\alpha} \bra{-}$:

\begin{equation}
    \embedfunctor\interpFbig[\  \begin{tikzpicture}
	\begin{pgfonlayer}{nodelayer}
		\node [style=Doubled X Phase dot] (0) at (-0.75, 0) {$\alpha$};
		\node [style=none] (1) at (1, 0) {};
	\end{pgfonlayer}
	\begin{pgfonlayer}{edgelayer}
		\draw [style=boldedge] (0) to (1.center);
	\end{pgfonlayer}
\end{tikzpicture}
 \ ] \quad = \quad
    \begin{pmatrix}
        1 \\ 0 \\ -\sin (\alpha) \\ \cos(\alpha)
    \end{pmatrix}
\end{equation}

as well as $R_Z$, $R_X$,

\begin{equation}
    \embedfunctor\interpFbig[\  \begin{tikzpicture}
	\begin{pgfonlayer}{nodelayer}
		\node [style=doubled z phase dot] (0) at (0, 0) {$\alpha$};
		\node [style=none] (1) at (1.75, 0) {};
		\node [style=none] (2) at (-1.75, 0) {};
	\end{pgfonlayer}
	\begin{pgfonlayer}{edgelayer}
		\draw [style=boldedge] (0) to (1.center);
		\draw [style=boldedge] (2.center) to (0);
	\end{pgfonlayer}
\end{tikzpicture}
 \ ]\quad = \quad
    \begin{pmatrix}
        1 & 0 & 0 & 0 \\
        0 & \cos(\alpha) & -\sin(\alpha) & 0 \\
        0 & \sin(\alpha) & \cos(\alpha) & 0 \\
        0 & 0 & 0 & 1
    \end{pmatrix}
\end{equation}

\begin{equation}
    \embedfunctor\interpFbig[\  \begin{tikzpicture}
	\begin{pgfonlayer}{nodelayer}
		\node [style=Doubled X Phase dot] (0) at (0, 0) {$\alpha$};
		\node [style=none] (1) at (1.75, 0) {};
		\node [style=none] (2) at (-1.75, 0) {};
	\end{pgfonlayer}
	\begin{pgfonlayer}{edgelayer}
		\draw [style=boldedge] (0) to (1.center);
		\draw [style=boldedge] (2.center) to (0);
	\end{pgfonlayer}
\end{tikzpicture}
 \ ]\quad = \quad
    \begin{pmatrix}
        1 & 0 & 0 & 0 \\
        0 & 1 & 0 & 0 \\
        0 & 0 & \cos(\alpha) & -\sin(\alpha) \\
        0 & 0 & \sin(\alpha) & \cos(\alpha)
    \end{pmatrix}
\end{equation}

Hadamard,

\begin{equation}
    \embedfunctor\interpFbig[\  \begin{tikzpicture}
	\begin{pgfonlayer}{nodelayer}
		\node [style=doubled hadamard] (0) at (0, 0) {};
		\node [style=none] (1) at (2, 0) {};
		\node [style=none] (2) at (-2, 0) {};
	\end{pgfonlayer}
	\begin{pgfonlayer}{edgelayer}
		\draw [style=boldedge] (0) to (1.center);
		\draw [style=boldedge] (0) to (2.center);
	\end{pgfonlayer}
\end{tikzpicture}
 \ ] \quad = \quad
    \begin{pmatrix}
        1 &  0 &  0 &  0 \\
        0 &  0 &  0 &  1 \\
        0 &  0 & -1 &  0 \\
        0 &  1 &  0 &  0
    \end{pmatrix}
\end{equation}
 and CNOT gates (where $I$ and $0$ are the $2\times 2$ identity and zero matrices, respectively):

\begin{equation}
    \embedfunctor\interpFbigg[\  \begin{tikzpicture}
	\begin{pgfonlayer}{nodelayer}
		\node [style=doubled z phase dot] (0) at (0, 1) {};
		\node [style=none] (1) at (1.75, 1) {};
		\node [style=none] (2) at (-1.75, 1) {};
		\node [style=Doubled X Phase dot] (3) at (0, -1) {};
		\node [style=none] (4) at (1.75, -1) {};
		\node [style=none] (5) at (-1.75, -1) {};
	\end{pgfonlayer}
	\begin{pgfonlayer}{edgelayer}
		\draw [style=boldedge] (0) to (1.center);
		\draw [style=boldedge] (2.center) to (0);
		\draw [style=boldedge] (3) to (4.center);
		\draw [style=boldedge] (5.center) to (3);
		\draw [style=boldedge] (0) to (3);
	\end{pgfonlayer}
\end{tikzpicture} \ ] \quad = \quad
    \begin{pmatrix}
        I & 0 & 0 & 0 & 0 & 0 & 0 & 0 \\
        0 & 0 & 0 & 0 & 0 & 0 & 0 & I \\
        0 & 0 & 0 & 0 & 0 & 0 & 0 & 0 \\
        0 & 0 & 0 & 0 & 0 & -i\sigma_{y} & 0 & 0 \\
        0 & 0 & 0 & 0 & \sigma_{x} & 0 & 0 & 0 \\
        0 & 0 & 0 & i\sigma_{y} & 0 & 0 & 0 & 0 \\
        0 & 0 & 0 & 0 & 0 & 0 & I & 0 \\
        0 & I & 0 & 0 & 0 & 0 & 0 & 0 \\
    \end{pmatrix}
\end{equation}

\begin{remark}
    This construction could in principle be extended to qudits by using generalizations of Pauli matrices, see \cite{kimura2003bloch} and \cite{bertlmann2008bloch}.
\end{remark}

Now that we have a functor to go from our abstract quantum circuit representations in $\CPMQ$ to concrete matrices in $\MatR$, it remains to go from $\MatR$ to $\Smooth$. This is done by the following functor:

 \begin{definition}
     The functor $\smoothfunctor: \MatR \longrightarrow \Smooth$ is defined by the following:
     \begin{description}[labelindent=2.5em]
        \item[Objects] $\smoothfunctor$ is identity on objects $n \mapsto n$.
        \item[Morphisms] $m \times n$ real matrices are mapped to the corresponding smooth map $\mathbb{R}^n \longrightarrow \mathbb{R}^m$.
     \end{description}
 \end{definition}

The composition of these two functors $\boxfunctor := \embedfunctor \fcmp \smoothfunctor$ gives us a functor from quantum to classical, which we will use to draw a functor box in our string diagrams.

 We now prove a theorem about our functor $\boxfunctor$, which will be of importance when constructing our quantum machine learning string diagrams.

We now prove that $\boxfunctor$ is lax monoidal.

\begin{theorem}
The functor $\boxfunctor$ is lax monoidal.
\end{theorem}

\begin{proof}
    We begin by defining the coherence maps:
    \begin{itemize}[leftmargin=25pt]
    \item The morphism $\epsilon: 0 \longrightarrow \boxfunctor(0) = 1$ is the map $\mathbb{R}^0 \rightarrow \mathbb{R}$ taking the single element $* \mapsto 1$.
    \item We define the natural transformation $\mu_{n, m}: \boxfunctor(n) + \boxfunctor(m) \longrightarrow \boxfunctor(n + m)$ by noting that in $\Smooth$, the object $\boxfunctor(n) + \boxfunctor(m) = 4^n + 4^m$ corresponds to the Euclidean space $\mathbb{R}^{4^n + 4^m}$. Elements of this space can be characterized by pairs of $4^n$ and $4^m$-dimensional real vectors $(\Vec{v},\ \Vec{w})$. We map these to the space $\mathbb{R}^{4^{n+m}}$ corresponding to $\boxfunctor(n + m) = 4^{n + m}$ as follows:
    \begin{equation}
        (\Vec{v},\ \Vec{w}) \longmapsto \Vec{v} \otimes \Vec{w}
    \end{equation}
    Where $\otimes$ denotes the Kronecker product.
    \end{itemize}

    We now prove that these maps satisfy the coherence conditions. For clarity, we denote the monoidal product on morphisms as $\otimes$ in $\CPMQ$ and as $\times$ in $\Smooth$, reflecting their operational behaviour.
    \begin{itemize}[leftmargin=25pt]
    \item For associativity, we must show that the following diagram commutes:
    \begin{equation}
    \begin{tikzcd}
        (\boxfunctor (n) + \boxfunctor (m)) + \boxfunctor (k) \arrow[r, "\alpha^{\Smooth}"] \arrow[d, "\mu_{n, m} \times id_{\boxfunctor(k)}", swap]
        & \boxfunctor (n) + (\boxfunctor (m) + \boxfunctor (k)) \arrow[d, "id_{\boxfunctor(n)} \times \mu_{m, k}"]
        \\ \boxfunctor (n + m) + \boxfunctor (k) \arrow[d, "\mu_{n + m, k}", swap]
        & \boxfunctor (n) + \boxfunctor (m + k) \arrow[d, "\mu_{n, m + k}"]
        \\ \boxfunctor ((n + m) + k) \arrow[r, "\boxfunctor (\alpha^{\CPMQ})", swap]
        & \boxfunctor (n + (m + k))
    \end{tikzcd}
    \end{equation}

    Where $\alpha$ denotes the associators in $\Smooth$ and $\CPMQ$. Given that the associators in $\Smooth$ and $\CPMQ$ are trivial, we draw the simpler diagram:
    \begin{equation}
    \begin{tikzcd}
        \boxfunctor (n) + \boxfunctor (m) + \boxfunctor (k) \arrow[r, "id_{\boxfunctor (n)} \times \mu_{m, k}"] \arrow[d, "\mu_{n, m} \times id_{\boxfunctor (k)}", swap]
        & \boxfunctor (n) + \boxfunctor (m + k) \arrow[d, "\mu_{n, m + k}"]
        \\ \boxfunctor (n + m) + \boxfunctor (k) \arrow[r, "\mu_{n + m, k}", swap]
        & \boxfunctor (n + m + k)
    \end{tikzcd}
    \end{equation}
    We prove this diagram commutes as follows. Assume any element $(\Vec{u},\ \Vec{v},\ \Vec{w})$ from the top left of the square. The top corner of the square maps this element as follows:
    \begin{equation}
        (\Vec{u},\ \Vec{v},\ \Vec{w}) \mapsto (\Vec{u},\ \Vec{v} \otimes \Vec{w}) \mapsto (\Vec{u} \otimes (\Vec{v} \otimes \Vec{w}))
    \end{equation}
    Conversely, the bottom corner of the square maps this element as follows:
    \begin{equation}
        (\Vec{u},\ \Vec{v},\ \Vec{w}) \mapsto (\Vec{u} \otimes \Vec{v},\ \Vec{w}) \mapsto ((\Vec{u} \otimes \Vec{v}) \otimes \Vec{w})
    \end{equation}
    These two mappings are equal by the associativity of the Kronecker product. Hence the square commutes.
    \item For unitality, we must show that the following diagrams commute:
    \begin{equation}
    \begin{tikzcd}
        0 + \boxfunctor (n) \arrow[r, "\epsilon \times id_{\boxfunctor(n)}"] \arrow[d, "l^{\Smooth}", swap] & \boxfunctor(0) + \boxfunctor (n) \arrow[d, "\mu_{0, n}"]
        \\ \boxfunctor (n) & \boxfunctor (0 + n) \arrow[l, "\boxfunctor (l^{\CPMQ})"]
    \end{tikzcd}
    \end{equation}
    \begin{equation}
    \begin{tikzcd}
        \boxfunctor (n) + 0 \arrow[r, "id_{\boxfunctor(n)} \times \epsilon"] \arrow[d, "r^{\Smooth}_{\boxfunctor(n)}", swap] & \boxfunctor (n) + \boxfunctor (0) \arrow[d, "\mu_{n, 0}"]
        \\ \boxfunctor (n) & \boxfunctor (n + 0) \arrow[l, "\boxfunctor (r^{\CPMQ}_n)"]
    \end{tikzcd}
    \end{equation}
    Where $l$ and $r$ denote left and right unitors, respectively.
    Take an element from the top left corner of the first square. Tracing the long route to the right, down, then left, we have:
    \begin{equation}
        (*,\ \Vec{u}) \mapsto (1,\ \Vec{u}) \mapsto 1 \otimes \Vec{u} \mapsto \Vec{u}
    \end{equation}
    The route straight down maps the element as:
    \begin{equation}
        (*,\ \Vec{u}) \mapsto \Vec{u}
    \end{equation}
    Thus the diagram commutes. The second diagram is similar. Along the right, down, left route:
    \begin{equation}
        (\Vec{u},\ *) \mapsto (\Vec{u},\ 1) \mapsto \Vec{u} \otimes 1 \mapsto \Vec{u}
    \end{equation}
    Along the route straight down:
    \begin{equation}
        (\Vec{u},\ *) \mapsto \Vec{u}
    \end{equation}
    \item It remains to show that $\mu$ is a natural transformation. For this we must prove that the following diagrams commute for all $n, m, k$ and completely positive superoperators $T, T'$ on the appropriate Hilbert spaces:
    \begin{equation}
    \begin{tikzcd}
        \boxfunctor (n) + \boxfunctor (m) \arrow[r, "\mu_{n, m}"] \arrow[d, "id_{\boxfunctor(n)} \times \boxfunctor (T)"] & \boxfunctor (n + m) \arrow[d, "\boxfunctor (id_n \otimes T)"]
        \\ \boxfunctor (n) + \boxfunctor (k) \arrow[r, "\mu_{n, k}"] & \boxfunctor (n + k)
    \end{tikzcd}
    \end{equation}
    and
    \begin{equation}
    \begin{tikzcd}
        \boxfunctor (n) + \boxfunctor (m) \arrow[r, "\mu_{n, m}"] \arrow[d, "\boxfunctor (T') \times id_{\boxfunctor(m)}"] & \boxfunctor (n + m) \arrow[d, "\boxfunctor (T' \otimes id_m)"]
        \\ \boxfunctor (k) + \boxfunctor (m) \arrow[r, "\mu_{n, k}"] & \boxfunctor (k + m)
    \end{tikzcd}
    \end{equation}
    We prove only that the first square commutes: the second is similar by the symmetry of the tensor and Cartesian products.
    
    First notice that since $\embedfunctor$ represents linear maps as matrices with respect to orthonormal bases, it preserves the tensor product: $\embedfunctor\left(T_1 \otimes T_2\right) = \embedfunctor\left(T_1\right) \otimes \embedfunctor\left(T_2\right)$\footnote{$\embedfunctor$ is in fact strict monoidal.}. Thus $\boxfunctor(id_n + T) = \smoothfunctor \fcmp \embedfunctor\left(id_n \otimes T\right) = \smoothfunctor(I_{4^n} \otimes M)$, where $I_{4^n}$ is the $4^n \times 4^n$ identity matrix and $M$ is the matrix for $T$ with respect to the tensored Pauli bases.

    Next note that since by definition of the tensor product, $(M_1 \otimes M_2)(\Vec{u} \otimes \Vec{v}) = M_1\Vec{u} \otimes M_2\Vec{v}$. Thus in order for the map $\smoothfunctor(M_1 \otimes M_2)$ in $\Smooth$ to correspond to $M_1 \otimes M_2$, it must act on $\Vec{u} \otimes \Vec{v}$, viewed as an element in Euclidean space, as $\smoothfunctor(M_1 \otimes M_2)(\Vec{u} \otimes \Vec{v}) = \smoothfunctor(M_1)(\Vec{u}) \otimes \smoothfunctor(M_2)(\Vec{v})$.
    
    With this in mind, when taking any element of from the top left corner of the first square and taking the upper route to the bottom right corner, we see that:
    \begin{equation}
        (\Vec{u},\ \Vec{v}) \mapsto \Vec{u} \otimes \Vec{v} \mapsto \Vec{u} \otimes \Vec{w}
    \end{equation}
    Along the bottom route, we have
    \begin{equation}
        (\Vec{u},\ \Vec{v}) \mapsto (\Vec{u},\ \Vec{w}) \mapsto \Vec{u} \otimes \Vec{w}
    \end{equation}
    Thus the square commutes.
    \item Hence have shown that there exist two coherence maps $\epsilon$ and $\mu$ satisfying the appropriate conditions to make $\boxfunctor$ lax monoidal.
    \end{itemize}
\end{proof}

For any two doubled ZX diagrams $D_1,\ D_2 \in \doublezx$, we have $\interpF[D_1 \otimes D_2] =  \interpF[D_1] \otimes \interpF[D_2]$. This leads to the following corollary:

\begin{corollary}
    The functor $\interpF \fcmp \boxfunctor$ is lax monoidal.
\end{corollary}

This allows us to use doubled ZX diagrams in our hybrid quantum-classical diagrams.

\begin{remark}
    In the context of quantum information, the fact that our functor is lax monoidal but not oplax is not a bug, but a feature. Classical information can be encoded into a quantum system accurately, reflected by the preservation of wires entering the quantum system in our string diagrams. On the other hand, extracting information from the quantum system involves measurement, collapsing the superposition of the quantum state. This yields a probabilistic outcome and is reflected by the bunching up of wires leaving the quantum system into scalar probabilities.
\end{remark}
\section{Application}

In this section we apply our theoretical framework to the quantum machine learning experiments outlined in \cite{farhi2018classification} for the MNIST\cite{MNIST} classification task. In these experiments, a simplified version of the MNIST dataset is used: a subset of two of the 10 digits are chosen, the $28 \times 28$-pixel input images are downsampled to $4 \times 4$ and binarized, in the sense that each pixel is then set to either $0$ or $1$ depending on whether it exceeds a threshold value. This leaves us with images represented by $16$-digit binary strings $\left( x_0,\ x_1,\ \ldots,\ x_{15} \right)$, which can be encoded on a quantum computer as states $\bra{x_0 x_1 \ldots x_{15}}$.

\subsection{Quantum Data Encoding}

For clarity of exposition, we will draw string diagrams for $4$-bit binary strings, from which the case of $16$-bit strings follows analogously. Using our hybrid classical-quantum framework, we can explicitly express the encoding of these classical binary string into the quantum system. Take, for example, the string $0100$. This would become the quantum state $\bra{0100}$ obtained by applying an $X$ gate to the second qubit, which we can express using the following ZX-diagram:

\begin{equation}\label{zx0100}
    \begin{tikzpicture}
	\begin{pgfonlayer}{nodelayer}
		\node [style=Doubled X Phase dot] (0) at (-1.25, 1.9) {};
		\node [style=Doubled X Phase dot] (1) at (-1.25, 0.65) {};
		\node [style=Doubled X Phase dot] (2) at (-1.25, -0.6) {};
		\node [style=Doubled X Phase dot] (3) at (0, 1.9) {};
		\node [style=Doubled X Phase dot] (4) at (0, 0.65) {$\pi$};
		\node [style=Doubled X Phase dot] (5) at (0, -0.6) {};
		\node [style=none] (6) at (1.25, 1.9) {};
		\node [style=none] (7) at (1.25, 0.65) {};
		\node [style=none] (8) at (1.25, -0.6) {};
		\node [style=Doubled X Phase dot] (9) at (-1.25, -1.85) {};
		\node [style=Doubled X Phase dot] (10) at (0, -1.85) {};
		\node [style=none] (11) at (1.25, -1.85) {};
	\end{pgfonlayer}
	\begin{pgfonlayer}{edgelayer}
		\draw [style=boldedge] (0) to (3);
		\draw [style=boldedge] (3) to (6.center);
		\draw [style=boldedge] (4) to (1);
		\draw [style=boldedge] (2) to (5);
		\draw [style=boldedge] (5) to (8.center);
		\draw [style=boldedge] (4) to (7.center);
		\draw [style=boldedge] (9) to (10);
		\draw [style=boldedge] (10) to (11.center);
	\end{pgfonlayer}
\end{tikzpicture}
\end{equation}

In our hybrid string diagrams, quantum systems are contained within functorial boxes\cite{mellies2006functorial}. We can thus turn (\ref{zx0100}) into a hybrid diagram representing the encoding of the classical data $0100$ into a quantum state by placing it inside a functor box and `pulling out' the spiders into the classical part of the diagram:

\begin{equation}
    \begin{tikzpicture}
	\begin{pgfonlayer}{nodelayer}
		\node [style=Doubled X Phase dot] (0) at (-1.25, 1.9) {};
		\node [style=Doubled X Phase dot] (1) at (-1.25, 0.65) {};
		\node [style=Doubled X Phase dot] (2) at (-1.25, -0.6) {};
		\node [style=Doubled X Phase dot] (3) at (0, 1.9) {};
		\node [style=Doubled X Phase dot] (4) at (0, 0.65) {$\pi$};
		\node [style=Doubled X Phase dot] (5) at (0, -0.6) {};
		\node [style=none] (6) at (1.25, 1.9) {};
		\node [style=none] (7) at (1.25, 0.65) {};
		\node [style=none] (8) at (1.25, -0.6) {};
		\node [style=Doubled X Phase dot] (9) at (-1.25, -1.85) {};
		\node [style=Doubled X Phase dot] (10) at (0, -1.85) {};
		\node [style=none] (11) at (1.25, -1.85) {};
		\node [style=none] (12) at (-2.25, 2.775) {};
		\node [style=none] (13) at (-2.25, -2.775) {};
		\node [style=none] (14) at (1.5, -2.775) {};
		\node [style=none] (15) at (1.5, 2.775) {};
	\end{pgfonlayer}
	\begin{pgfonlayer}{edgelayer}
		\draw [style=boldedge] (0) to (3);
		\draw [style=boldedge] (3) to (6.center);
		\draw [style=boldedge] (4) to (1);
		\draw [style=boldedge] (2) to (5);
		\draw [style=boldedge] (5) to (8.center);
		\draw [style=boldedge] (4) to (7.center);
		\draw [style=boldedge] (9) to (10);
		\draw [style=boldedge] (10) to (11.center);
		\draw [style=functorbox] (15.center)
			 to (12.center)
			 to (13.center)
			 to (14.center);
	\end{pgfonlayer}
\end{tikzpicture} \quad = \quad \begin{tikzpicture}
	\begin{pgfonlayer}{nodelayer}
		\node [style=Doubled X Phase dot] (3) at (-1.25, 1.9) {};
		\node [style=Doubled X Phase dot] (4) at (-1.25, 0.65) {$\pi$};
		\node [style=Doubled X Phase dot] (5) at (-1.25, -0.6) {};
		\node [style=none] (6) at (1.25, 1.9) {};
		\node [style=none] (7) at (1.25, 0.65) {};
		\node [style=none] (8) at (1.25, -0.6) {};
		\node [style=Doubled X Phase dot] (10) at (-1.25, -1.85) {};
		\node [style=none] (11) at (1.25, -1.85) {};
		\node [style=none] (12) at (-2.25, 2.775) {};
		\node [style=none] (13) at (-2.25, -2.775) {};
		\node [style=none] (14) at (1.5, -2.775) {};
		\node [style=none] (15) at (1.5, 2.775) {};
	\end{pgfonlayer}
	\begin{pgfonlayer}{edgelayer}
		\draw [style=boldedge] (3) to (6.center);
		\draw [style=boldedge] (5) to (8.center);
		\draw [style=boldedge] (4) to (7.center);
		\draw [style=boldedge] (10) to (11.center);
		\draw [style=functorbox] (15.center)
			 to (12.center)
			 to (13.center)
			 to (14.center);
	\end{pgfonlayer}
\end{tikzpicture} \quad
    = \quad \begin{tikzpicture}
	\begin{pgfonlayer}{nodelayer}
		\node [style=Doubled X Phase dot] (3) at (0, 2.625) {};
		\node [style=Doubled X Phase dot] (4) at (0, 0.875) {};
		\node [style=Doubled X Phase dot] (5) at (0, -0.875) {};
		\node [style=none] (6) at (1.5, 2.625) {};
		\node [style=none] (7) at (1.5, 0.875) {};
		\node [style=none] (8) at (1.5, -0.875) {};
		\node [style=none] (9) at (-1.25, 3.375) {};
		\node [style=none] (10) at (-1.25, -3.275) {};
		\node [style=none] (11) at (1.5, -3.275) {};
		\node [style=none] (12) at (1.5, 3.375) {};
		\node [style=none] (13) at (-2.25, 2.625) {};
		\node [style=none] (26) at (-2.25, -0.875) {};
		\node [style=none] (27) at (-2.25, 2.625) {};
		\node [style=none] (28) at (-2.25, 0.875) {};
		\node [style=none] (29) at (-2.25, 3.275) {};
		\node [style=none] (30) at (-2.25, 1.975) {};
		\node [style=none] (31) at (-3.15, 2.625) {};
		\node [style=none] (32) at (-2.55, 2.625) {$0$};
		\node [style=none] (33) at (-2.25, -0.875) {};
		\node [style=none] (34) at (-2.25, 1.475) {};
		\node [style=none] (35) at (-2.25, 0.275) {};
		\node [style=none] (36) at (-3.1, 0.875) {};
		\node [style=none] (37) at (-2.25, -0.875) {};
		\node [style=none] (38) at (-2.25, -0.275) {};
		\node [style=none] (39) at (-2.25, -1.475) {};
		\node [style=none] (40) at (-3.1, -0.875) {};
		\node [style=none] (41) at (-2.55, -0.875) {$0$};
		\node [style=none] (42) at (-2.55, 0.875) {$1$};
		\node [style=none] (43) at (-1.25, 2.625) {};
		\node [style=none] (44) at (-1.25, 0.875) {};
		\node [style=none] (45) at (-1.25, -0.875) {};
		\node [style=Doubled X Phase dot] (46) at (0, -2.475) {};
		\node [style=none] (47) at (1.5, -2.475) {};
		\node [style=none] (48) at (-2.25, -2.475) {};
		\node [style=none] (49) at (-2.25, -2.475) {};
		\node [style=none] (50) at (-2.25, -2.475) {};
		\node [style=none] (51) at (-2.25, -1.875) {};
		\node [style=none] (52) at (-2.25, -3.075) {};
		\node [style=none] (53) at (-3.1, -2.475) {};
		\node [style=none] (54) at (-2.55, -2.475) {$0$};
		\node [style=none] (55) at (-1.25, -2.475) {};
	\end{pgfonlayer}
	\begin{pgfonlayer}{edgelayer}
		\draw [style=boldedge] (3) to (6.center);
		\draw [style=boldedge] (5) to (8.center);
		\draw [style=boldedge] (4) to (7.center);
		\draw [style=functorbox] (12.center)
			 to (9.center)
			 to (10.center)
			 to (11.center);
		\draw (29.center) to (31.center);
		\draw (31.center) to (30.center);
		\draw (30.center) to (29.center);
		\draw (34.center) to (36.center);
		\draw (36.center) to (35.center);
		\draw (35.center) to (34.center);
		\draw (38.center) to (40.center);
		\draw (40.center) to (39.center);
		\draw (39.center) to (38.center);
		\draw [style=boldedge] (45.center) to (5);
		\draw [style=boldedge] (44.center) to (4);
		\draw [style=boldedge] (43.center) to (3);
		\draw (27.center) to (43.center);
		\draw (37.center) to (45.center);
		\draw (28.center) to (44.center);
		\draw [style=boldedge] (46) to (47.center);
		\draw (51.center) to (53.center);
		\draw (53.center) to (52.center);
		\draw (52.center) to (51.center);
		\draw [style=boldedge] (55.center) to (46);
		\draw (50.center) to (55.center);
	\end{pgfonlayer}
\end{tikzpicture}
\end{equation}

\subsection{Quantum Measurement}

The experiment in \cite{farhi2018classification} distinguishes between two different MNIST digit classes by measuring a Pauli operator, such as $\sigma_Z$, on single classical readout qubit. We can obtain this particular expectation value from the probability of measuring the state $\bra{0}$ on the readout qubit via $\expectation{\sigma_Z} = \mathrm{Prob}(\bra{0}) - \mathrm{Prob}(\bra{1}) = 2\mathrm{Prob}(\bra{0}) - 1$. Measuring this probability on, for example, the last qubit of a 3-qubit system can be represented as the following doubled ZX diagram, where we make use of the discard operation provided by the CPM construction:
\begin{equation}
    \begin{tikzpicture}
	\begin{pgfonlayer}{nodelayer}
		\node [style=Doubled X Phase dot] (2) at (1.25, -1.75) {};
		\node [style=none] (5) at (-1, -1.75) {};
		\node [style=none] (9) at (-1, 0) {};
		\node [style=none] (10) at (1, 0) {};
		\node [style=none] (11) at (1, 0.65) {};
		\node [style=none] (12) at (1, -0.65) {};
		\node [style=none] (13) at (1.25, 0.4) {};
		\node [style=none] (14) at (1.25, -0.4) {};
		\node [style=none] (15) at (1.5, 0.25) {};
		\node [style=none] (16) at (1.5, -0.25) {};
		\node [style=none] (17) at (-1, 1.75) {};
		\node [style=none] (18) at (1, 1.75) {};
		\node [style=none] (19) at (1, 2.4) {};
		\node [style=none] (20) at (1, 1.1) {};
		\node [style=none] (21) at (1.25, 2.15) {};
		\node [style=none] (22) at (1.25, 1.35) {};
		\node [style=none] (23) at (1.5, 2) {};
		\node [style=none] (24) at (1.5, 1.5) {};
	\end{pgfonlayer}
	\begin{pgfonlayer}{edgelayer}
		\draw [style=boldedge] (2) to (5.center);
		\draw [style=boldedge] (11.center) to (12.center);
		\draw [style=boldedge] (13.center) to (14.center);
		\draw [style=boldedge] (15.center) to (16.center);
		\draw [style=boldedge] (10.center) to (9.center);
		\draw [style=boldedge] (19.center) to (20.center);
		\draw [style=boldedge] (21.center) to (22.center);
		\draw [style=boldedge] (23.center) to (24.center);
		\draw [style=boldedge] (18.center) to (17.center);
	\end{pgfonlayer}
\end{tikzpicture}
\end{equation}
When contained within a functor box in our hybrid diagrams, this measurement outputs a single wire into our classical system representing a single real number: the measurement probability.
\begin{equation}
    \begin{tikzpicture}
	\begin{pgfonlayer}{nodelayer}
		\node [style=Doubled X Phase dot] (2) at (0, -1.75) {};
		\node [style=none] (5) at (-2.25, -1.75) {};
		\node [style=none] (9) at (-2.25, 0) {};
		\node [style=none] (10) at (-0.25, 0) {};
		\node [style=none] (11) at (-0.25, 0.65) {};
		\node [style=none] (12) at (-0.25, -0.65) {};
		\node [style=none] (13) at (0, 0.4) {};
		\node [style=none] (14) at (0, -0.4) {};
		\node [style=none] (15) at (0.25, 0.25) {};
		\node [style=none] (16) at (0.25, -0.25) {};
		\node [style=none] (17) at (-2.25, 1.75) {};
		\node [style=none] (18) at (-0.25, 1.75) {};
		\node [style=none] (19) at (-0.25, 2.4) {};
		\node [style=none] (20) at (-0.25, 1.1) {};
		\node [style=none] (21) at (0, 2.15) {};
		\node [style=none] (22) at (0, 1.35) {};
		\node [style=none] (23) at (0.25, 2) {};
		\node [style=none] (24) at (0.25, 1.5) {};
		\node [style=none] (25) at (-2.5, 3) {};
		\node [style=none] (26) at (-2.75, -3) {};
		\node [style=none] (27) at (1.25, -3) {};
		\node [style=none] (28) at (1.25, 3) {};
		\node [style=none] (29) at (1.25, 0) {};
		\node [style=none] (30) at (3.75, 0) {};
	\end{pgfonlayer}
	\begin{pgfonlayer}{edgelayer}
		\draw [style=boldedge] (2) to (5.center);
		\draw [style=boldedge] (11.center) to (12.center);
		\draw [style=boldedge] (13.center) to (14.center);
		\draw [style=boldedge] (15.center) to (16.center);
		\draw [style=boldedge] (10.center) to (9.center);
		\draw [style=boldedge] (19.center) to (20.center);
		\draw [style=boldedge] (21.center) to (22.center);
		\draw [style=boldedge] (23.center) to (24.center);
		\draw [style=boldedge] (18.center) to (17.center);
		\draw [style=functorbox] (25.center)
			 to (28.center)
			 to (27.center)
			 to (26.center);
		\draw (29.center) to (30.center);
	\end{pgfonlayer}
\end{tikzpicture}
\end{equation}

\subsection{Quantum Machine Learning with String Diagrams}

We now present a full hybrid-quantum classical circuit for the experiments in \cite{farhi2018classification} using our framework. The parameterized part of the quantum circuits used consist of $ZX$ and $XX$ interactions, which can be expressed succinctly in the ZX-calculus using `phase gadgets'\cite{cowtan2020phase}:

\begin{equation}
    \scalebox{0.8}{\begin{tikzpicture}
	\begin{pgfonlayer}{nodelayer}
		\node [style=none] (11) at (9, 3.75) {};
		\node [style=none] (13) at (9, -0.25) {};
		\node [style=none] (15) at (9, 1.75) {};
		\node [style=doubled z phase dot] (23) at (-5.25, 3.75) {};
		\node [style=doubled z phase dot] (24) at (0.75, -0.25) {};
		\node [style=Doubled X Phase dot] (68) at (-8, 3.75) {$x_1$};
		\node [style=Doubled X Phase dot] (69) at (-8, 1.75) {$x_2$};
		\node [style=Doubled X Phase dot] (70) at (-8, -0.25) {$x_3$};
		\node [style=none] (76) at (9, -2.25) {};
		\node [style=Doubled X Phase dot] (84) at (-8, -2.25) {$x_4$};
		\node [style=doubled z phase dot] (89) at (-2.25, 1.75) {};
		\node [style=doubled z phase dot] (91) at (3.75, -2.25) {};
		\node [style=Doubled X Phase dot] (92) at (-3.75, 2.75) {};
		\node [style=Doubled X Phase dot] (93) at (-0.75, 0.75) {};
		\node [style=Doubled X Phase dot] (94) at (2.25, -1.25) {};
		\node [style=doubled z phase dot] (95) at (-0.75, 4.75) {$\theta_2$};
		\node [style=doubled z phase dot] (96) at (-3.75, 4.75) {$\theta_1$};
		\node [style=doubled z phase dot] (97) at (2.25, 4.75) {$\theta_3$};
		\node [style=Doubled X Phase dot] (100) at (-8, -4.25) {$\pi$};
		\node [style=doubled z phase dot] (101) at (3.75, -4.25) {};
		\node [style=doubled z phase dot] (102) at (-5.25, -4.25) {};
		\node [style=doubled z phase dot] (103) at (-2.25, -4.25) {};
		\node [style=doubled z phase dot] (104) at (0.75, -4.25) {};
		\node [style=Doubled X Phase dot] (105) at (5.25, -3.25) {};
		\node [style=doubled z phase dot] (106) at (5.25, 4.75) {$\theta_4$};
		\node [style=none] (115) at (9, 3.75) {};
		\node [style=none] (116) at (9, 4.4) {};
		\node [style=none] (117) at (9, 3.1) {};
		\node [style=none] (118) at (9.25, 4.15) {};
		\node [style=none] (119) at (9.25, 3.35) {};
		\node [style=none] (120) at (9.5, 4) {};
		\node [style=none] (121) at (9.5, 3.5) {};
		\node [style=none] (122) at (9, 1.75) {};
		\node [style=none] (123) at (9, 2.4) {};
		\node [style=none] (124) at (9, 1.1) {};
		\node [style=none] (125) at (9.25, 2.15) {};
		\node [style=none] (126) at (9.25, 1.35) {};
		\node [style=none] (127) at (9.5, 2) {};
		\node [style=none] (128) at (9.5, 1.5) {};
		\node [style=none] (129) at (9, -0.25) {};
		\node [style=none] (130) at (9, 0.4) {};
		\node [style=none] (131) at (9, -0.9) {};
		\node [style=none] (132) at (9.25, 0.15) {};
		\node [style=none] (133) at (9.25, -0.65) {};
		\node [style=none] (134) at (9.5, 0) {};
		\node [style=none] (135) at (9.5, -0.5) {};
		\node [style=none] (136) at (9, -2.25) {};
		\node [style=none] (137) at (9, -1.6) {};
		\node [style=none] (138) at (9, -2.9) {};
		\node [style=none] (139) at (9.25, -1.85) {};
		\node [style=none] (140) at (9.25, -2.65) {};
		\node [style=none] (141) at (9.5, -2) {};
		\node [style=none] (142) at (9.5, -2.5) {};
		\node [style=doubled hadamard] (153) at (-6.75, -4.25) {};
		\node [style=Doubled X Phase dot] (155) at (9.25, -4.25) {};
		\node [style=doubled z phase dot] (156) at (6.5, -4.25) {$-\frac{\pi}{2}$};
		\node [style=doubled hadamard] (157) at (8, -4.25) {};
		\node [style=doubled hadamard] (158) at (5, -4.25) {};
	\end{pgfonlayer}
	\begin{pgfonlayer}{edgelayer}
		\draw [style=boldedge] (91) to (76.center);
		\draw [style=boldedge] (89) to (15.center);
		\draw [style=boldedge] (96) to (92);
		\draw [style=boldedge] (95) to (93);
		\draw [style=boldedge] (97) to (94);
		\draw [style=boldedge] (68) to (23);
		\draw [style=boldedge] (84) to (91);
		\draw [style=boldedge] (24) to (70);
		\draw [style=boldedge] (69) to (89);
		\draw [style=boldedge] (23) to (11.center);
		\draw [style=boldedge] (24) to (13.center);
		\draw [style=boldedge] (104) to (101);
		\draw [style=boldedge] (104) to (103);
		\draw [style=boldedge] (102) to (103);
		\draw [style=boldedge] (106) to (105);
		\draw [style=boldedge] (92) to (23);
		\draw [style=boldedge] (92) to (102);
		\draw [style=boldedge] (93) to (89);
		\draw [style=boldedge] (93) to (103);
		\draw [style=boldedge] (94) to (24);
		\draw [style=boldedge] (94) to (104);
		\draw [style=boldedge] (105) to (91);
		\draw [style=boldedge] (105) to (101);
		\draw [style=boldedge] (116.center) to (117.center);
		\draw [style=boldedge] (118.center) to (119.center);
		\draw [style=boldedge] (120.center) to (121.center);
		\draw [style=boldedge] (123.center) to (124.center);
		\draw [style=boldedge] (125.center) to (126.center);
		\draw [style=boldedge] (127.center) to (128.center);
		\draw [style=boldedge] (130.center) to (131.center);
		\draw [style=boldedge] (132.center) to (133.center);
		\draw [style=boldedge] (134.center) to (135.center);
		\draw [style=boldedge] (137.center) to (138.center);
		\draw [style=boldedge] (139.center) to (140.center);
		\draw [style=boldedge] (141.center) to (142.center);
		\draw [style=boldedge] (100) to (153);
		\draw [style=boldedge] (153) to (102);
		\draw [style=boldedge] (157) to (155);
		\draw [style=boldedge] (157) to (156);
		\draw [style=boldedge] (158) to (156);
		\draw [style=boldedge] (158) to (101);
	\end{pgfonlayer}
\end{tikzpicture}}
\end{equation}

Notice that a y-basis measurement is taken on the readout qubit, which was easy to express in ZX diagrams. We can integrate this circuit into our framework by placing it inside a functor box, then composing the output with classical computations for inferring the expectation value and computing a loss function:

\begin{equation}
    \scalebox{0.8}{\begin{tikzpicture}
	\begin{pgfonlayer}{nodelayer}
		\node [style=none] (11) at (3.75, 3.75) {};
		\node [style=none] (13) at (3.75, -0.25) {};
		\node [style=none] (15) at (3.75, 1.75) {};
		\node [style=none] (17) at (5.5, 6) {};
		\node [style=none] (18) at (-14.25, 6) {};
		\node [style=none] (19) at (-14.25, -5.75) {};
		\node [style=none] (20) at (5.5, -5.75) {};
		\node [style=none] (22) at (7.75, 1) {};
		\node [style=doubled z phase dot] (23) at (-10.25, 3.75) {};
		\node [style=doubled z phase dot] (24) at (-4.25, -0.25) {};
		\node [style=none] (56) at (7.75, 2.75) {};
		\node [style=none] (57) at (7.75, -0.75) {};
		\node [style=none] (60) at (9.9, 1) {$\lambda x.(2x - 1)$};
		\node [style=none] (61) at (12, 2.75) {};
		\node [style=none] (62) at (12, -0.75) {};
		\node [style=none] (63) at (12, 1) {};
		\node [style=Doubled X Phase dot] (68) at (-13, 3.75) {$x_1$};
		\node [style=Doubled X Phase dot] (69) at (-13, 1.75) {$x_2$};
		\node [style=Doubled X Phase dot] (70) at (-13, -0.25) {$x_3$};
		\node [style=none] (76) at (3.75, -2.25) {};
		\node [style=Doubled X Phase dot] (84) at (-13, -2.25) {$x_4$};
		\node [style=none] (87) at (5.5, 1) {};
		\node [style=doubled z phase dot] (89) at (-7.25, 1.75) {};
		\node [style=doubled z phase dot] (91) at (-1.25, -2.25) {};
		\node [style=Doubled X Phase dot] (92) at (-8.75, 2.75) {};
		\node [style=Doubled X Phase dot] (93) at (-5.75, 0.75) {};
		\node [style=Doubled X Phase dot] (94) at (-2.75, -1.25) {};
		\node [style=doubled z phase dot] (95) at (-5.75, 4.75) {$\theta_2$};
		\node [style=doubled z phase dot] (96) at (-8.75, 4.75) {$\theta_1$};
		\node [style=doubled z phase dot] (97) at (-2.75, 4.75) {$\theta_3$};
		\node [style=Doubled X Phase dot] (100) at (-13, -4.25) {$\pi$};
		\node [style=doubled z phase dot] (101) at (-1.25, -4.25) {};
		\node [style=doubled z phase dot] (102) at (-10.25, -4.25) {};
		\node [style=doubled z phase dot] (103) at (-7.25, -4.25) {};
		\node [style=doubled z phase dot] (104) at (-4.25, -4.25) {};
		\node [style=Doubled X Phase dot] (105) at (0.25, -3.25) {};
		\node [style=doubled z phase dot] (106) at (0.25, 4.75) {$\theta_4$};
		\node [style=none] (107) at (14.25, 1) {};
		\node [style=none] (108) at (14.25, 2.25) {};
		\node [style=none] (109) at (14.25, -1.75) {};
		\node [style=none] (110) at (15.15, 0.25) {$L$};
		\node [style=none] (111) at (16, 2.25) {};
		\node [style=none] (112) at (16, -1.75) {};
		\node [style=none] (113) at (16, 0.25) {};
		\node [style=none] (114) at (18, 0.25) {};
		\node [style=none] (115) at (3.75, 3.75) {};
		\node [style=none] (116) at (3.75, 4.4) {};
		\node [style=none] (117) at (3.75, 3.1) {};
		\node [style=none] (118) at (4, 4.15) {};
		\node [style=none] (119) at (4, 3.35) {};
		\node [style=none] (120) at (4.25, 4) {};
		\node [style=none] (121) at (4.25, 3.5) {};
		\node [style=none] (122) at (3.75, 1.75) {};
		\node [style=none] (123) at (3.75, 2.4) {};
		\node [style=none] (124) at (3.75, 1.1) {};
		\node [style=none] (125) at (4, 2.15) {};
		\node [style=none] (126) at (4, 1.35) {};
		\node [style=none] (127) at (4.25, 2) {};
		\node [style=none] (128) at (4.25, 1.5) {};
		\node [style=none] (129) at (3.75, -0.25) {};
		\node [style=none] (130) at (3.75, 0.4) {};
		\node [style=none] (131) at (3.75, -0.9) {};
		\node [style=none] (132) at (4, 0.15) {};
		\node [style=none] (133) at (4, -0.65) {};
		\node [style=none] (134) at (4.25, 0) {};
		\node [style=none] (135) at (4.25, -0.5) {};
		\node [style=none] (136) at (3.75, -2.25) {};
		\node [style=none] (137) at (3.75, -1.6) {};
		\node [style=none] (138) at (3.75, -2.9) {};
		\node [style=none] (139) at (4, -1.85) {};
		\node [style=none] (140) at (4, -2.65) {};
		\node [style=none] (141) at (4.25, -2) {};
		\node [style=none] (142) at (4.25, -2.5) {};
		\node [style=none] (143) at (14.25, -0.5) {};
		\node [style=none] (145) at (10.25, -2.75) {};
		\node [style=none] (146) at (10.25, -1.75) {};
		\node [style=none] (147) at (10.25, -3.75) {};
		\node [style=none] (148) at (9, -2.75) {};
		\node [style=none] (149) at (9.75, -2.75) {};
		\node [style=none] (150) at (9.75, -2.75) {$y$};
		\node [style=none] (151) at (11.75, -2.75) {};
		\node [style=doubled hadamard] (153) at (-11.75, -4.25) {};
		\node [style=Doubled X Phase dot] (154) at (4.25, -4.25) {};
		\node [style=doubled z phase dot] (155) at (1.5, -4.25) {$-\frac{\pi}{2}$};
		\node [style=doubled hadamard] (156) at (3, -4.25) {};
		\node [style=doubled hadamard] (157) at (0, -4.25) {};
	\end{pgfonlayer}
	\begin{pgfonlayer}{edgelayer}
		\draw [style=functorbox] (18.center)
			 to (17.center)
			 to (20.center)
			 to (19.center)
			 to cycle;
		\draw (57.center) to (56.center);
		\draw (57.center) to (62.center);
		\draw (56.center) to (61.center);
		\draw (61.center) to (62.center);
		\draw (87.center) to (22.center);
		\draw [style=boldedge] (91) to (76.center);
		\draw [style=boldedge] (89) to (15.center);
		\draw [style=boldedge] (96) to (92);
		\draw [style=boldedge] (95) to (93);
		\draw [style=boldedge] (97) to (94);
		\draw [style=boldedge] (68) to (23);
		\draw [style=boldedge] (84) to (91);
		\draw [style=boldedge] (24) to (70);
		\draw [style=boldedge] (69) to (89);
		\draw [style=boldedge] (23) to (11.center);
		\draw [style=boldedge] (24) to (13.center);
		\draw [style=boldedge] (104) to (101);
		\draw [style=boldedge] (104) to (103);
		\draw [style=boldedge] (102) to (103);
		\draw [style=boldedge] (106) to (105);
		\draw [style=boldedge] (92) to (23);
		\draw [style=boldedge] (92) to (102);
		\draw [style=boldedge] (93) to (89);
		\draw [style=boldedge] (93) to (103);
		\draw [style=boldedge] (94) to (24);
		\draw [style=boldedge] (94) to (104);
		\draw [style=boldedge] (105) to (91);
		\draw [style=boldedge] (105) to (101);
		\draw (109.center) to (108.center);
		\draw (109.center) to (112.center);
		\draw (108.center) to (111.center);
		\draw (111.center) to (112.center);
		\draw (113.center) to (114.center);
		\draw [style=boldedge] (116.center) to (117.center);
		\draw [style=boldedge] (118.center) to (119.center);
		\draw [style=boldedge] (120.center) to (121.center);
		\draw [style=boldedge] (123.center) to (124.center);
		\draw [style=boldedge] (125.center) to (126.center);
		\draw [style=boldedge] (127.center) to (128.center);
		\draw [style=boldedge] (130.center) to (131.center);
		\draw [style=boldedge] (132.center) to (133.center);
		\draw [style=boldedge] (134.center) to (135.center);
		\draw [style=boldedge] (137.center) to (138.center);
		\draw [style=boldedge] (139.center) to (140.center);
		\draw [style=boldedge] (141.center) to (142.center);
		\draw (146.center) to (148.center);
		\draw (148.center) to (147.center);
		\draw (147.center) to (146.center);
		\draw (145.center) to (151.center);
		\draw [in=180, out=0, looseness=1.25] (151.center) to (143.center);
		\draw (63.center) to (107.center);
		\draw [style=boldedge] (100) to (153);
		\draw [style=boldedge] (153) to (102);
		\draw [style=boldedge] (156) to (154);
		\draw [style=boldedge] (156) to (155);
		\draw [style=boldedge] (157) to (155);
		\draw [style=boldedge] (157) to (101);
	\end{pgfonlayer}
\end{tikzpicture}}
\end{equation}

Where $y$ represents the label indicating the true class of the input image. We can now pull the relevant spiders out of the box to obtain a diagram that captures the semantics of the hybrid computation by drawing a distinction between classical and quantum parts of the algorithm:

\begin{equation}
    \scalebox{0.8}{\input{./figures/qml_3.tikz}}
\end{equation}

\section{Conclusion and Future Work}

We have taken initial steps towards a category-theoretic formalism of hybrid quantum-classical machine learning. This was done by embedding a the dagger compact closed category $\CPM$, in which quantum computation takes place, into a Cartesian category $\Smooth$ for the classical computation. This was achieved using a lax monoidal functor, which we defined concretely as the composition of two functors: one from  $\CPM \longrightarrow \MatR$, and another from $\MatR \longrightarrow \Smooth$. Furthermore, through the use of functor boxes we were able to capture both the classical and quantum part of hybrid algorithms in a single string diagram where functor boxes delineate the quantum-classical interface. A notable consequence of the lax monoidality of our functor was that multiple wires could enter the quantum system when encoding classical data on the quantum computer whereas only one wire, representing a measurement probability, could exit the quantum system. The quantum part of the algorithm could be described as a ZX-diagram, making amenable to quantum circuit optimization methods. Finally, we applied our framework to the MNIST classification task in \cite{farhi2018classification}, giving us a denotational semantics of a concrete quantum machine learning algorithm that could form the basis of a programming language.

In future work, we aim to formalise parameterization and backpropagation in these hybrid algorithms using the parametric lens construction of \cite{cruttwell2024deep}. This will allow for a full specification of a hybrid quantum classical machine learning algorithm with string diagrams capturing both the forward and backward pass. In such a framework, we would need to specify the gradient of quantum parts of the algorithm. This can be done using gradient recipes: quantum circuits that compute the gradients of others. We note the use of the ZXW calculus for computing gradient recipes\cite{koch2022quantum}\cite{wang2022differentiating} and are interested in employing them in our framework.

\section{Acknowledgements}

Alexander would like to thank Simon Harrison for the generous support he receives from the Wolfson Harrison UK Research Council Quantum Foundation Scholarship.

\bibliographystyle{eptcs}
\bibliography{main}
\pagebreak
\appendix

\section{Appendix}

\subsection{Categories and Functors}\label{appendixCT}

\begin{definition}[Category]
A category $\mathcal{C}$ consists of the following:
\begin{enumerate}[label=(\roman*)]
    \item A collection\footnote{We use collections to avoid Russell's paradox, since objects may themselves be sets.} of \textit{objects} $\mathrm{Ob}(\mathcal{C})$.
    \item For each pair of objects $A, B$ a set $\mathcal{C}(A, B)$ whose elements are called \textit{morphisms} from $A$ to $B$.
    \item For each object $A$, a morphism $id_A \in \mathcal{C}(C, C)$ called the \textit{identity morphism} for $A$.
    \item For any three objects $A, B, C$ and any two morphisms $f \in \mathcal{C}(A, B)$ and $g \in \mathcal{C}(B, C)$, a morphism $f \fcmp g \in \mathcal{C}(A, C)$ called the \textit{composite} of $f$ and $g$.
\end{enumerate}
For notational convenience, we will often objects $C \in \mathrm{Ob}(\mathcal{C})$ as $C \in \mathcal{C}$, and morphisms $f \in \mathcal{C}(A, B)$ as $f: A \rightarrow B$.
These constituents are subject to the following conditions:
\begin{enumerate}[label=(\roman*)]
    \item For any morphism $f: A \rightarrow B$, composition with the identity gives the same morphism:
    \begin{equation*}
        id_A \fcmp f = f = f \fcmp id_B
    \end{equation*}
    \item For any three morphisms $f: A \rightarrow B$, $g: B \rightarrow C$, $h: C \rightarrow D$ their composition is associative:
    \begin{equation*}
        (f \fcmp g) \fcmp h = f 
        \fcmp (g \fcmp h)    
    \end{equation*}
    Given this equality, we will write such compositions simply as $f \fcmp g \fcmp h$.
\end{enumerate}
\end{definition}

\begin{definition}[Functor]
    A \textit{functor} $\mathrm{F}$ between two categories $\mathcal{C}$ and $\mathcal{D}$ consists of:
    \begin{enumerate}[label=(\roman*)]
        \item For each object $C \in \mathcal{C}$, an object $\mathrm{F}(C) \in \mathcal{D}$
        \item For each morphism $f: A \rightarrow B$ in $\mathcal{C}$, a morphism $\mathrm{F}(f): \mathrm{F}(A) \rightarrow \mathrm{F}(B)$ in $\mathcal{D}$.
    \end{enumerate}
    Subject to the following conditions:
    \begin{enumerate}[label=(\roman*)]
        \item Identities are preserved: $\mathrm{F}(id_C) = id_{\mathrm{F}(C)}$.
        \item Composition is preserved: $\mathrm{F}(f \fcmp g) = \mathrm{F}(f) \fcmp \mathrm{F}(g)$.
    \end{enumerate}
\end{definition}

\begin{definition}[Natural transformation]
Let $\mathrm{F}$ and $\mathrm{G}$ be functors $\mathcal{C} \rightarrow \mathcal{D}$. A \textit{natural transformation} $\alpha$ is a family of morphisms $\{ \alpha_C \}_{C \in \mathcal{C}}$ in $\mathcal{D}$, such that:
\begin{enumerate}[label=(\roman*)]
    \item The family contains exactly one morphism $\alpha_C$ for each object $C \in \mathcal{C}$, which we call \textit{the component of $\alpha$ at $C$}.
    \item For every morphism $f: A \rightarrow B$ in $\mathcal{C}$, the following diagram in $\mathcal{D}$ commutes:
\begin{equation*}
    \begin{tikzcd}
        F(A) \arrow[r, "\alpha_A"] \arrow[d, "F(f)", swap]
        & G(A) \arrow[d, "G(f)"]
        \\ F(B) \arrow[r, "\alpha_B"]
        & G(B)
    \end{tikzcd}
\end{equation*}
\end{enumerate}
\end{definition}

\begin{definition}[Isomorphism]
    A morphism $f: A \rightarrow B$ in a category $\mathcal{C}$ is an \textit{isomorphism} if there exists a morphism $f^{-1}: B \rightarrow A$ such that $f \fcmp f^{-1} = id_A$ and $f^{-1} \fcmp f = id_B$. The morphism $f^{-1}$ is referred to as the \textit{inverse} of $f$.
\end{definition}

\begin{definition}[Natural isomorphism]
A natural transformation $\alpha$ is called a \textit{natural isomorphism} if each component $\alpha_c$ is an isomorphism in $\mathcal{D}$.
\end{definition}

Product categories extend the notion of the Cartesian product of sets to categories.

\subsection{Symmetric Monoidal Categories and PROPs}

\begin{definition}[Product category]
    Let $\mathcal{C}$ and $\mathcal{D}$ be categories. The \textit{product category} $\mathcal{C} \times \mathcal{D}$ has:
    \begin{enumerate}[label=(\roman*)]
        \item As objects pairs $(C, D)$ of objects $C \in \mathcal{C}$ and $D \in \mathcal{D}$.
        \item as morphisms $(C_1, D_1) \rightarrow (C_2, D_2)$ pairs $(f, g)$, where $f: C_1 \rightarrow C_2$ is a morphism in $\mathcal{C}$ and $g: D_1 \rightarrow D_2$ is a morphism in $\mathcal{D}$.
    \end{enumerate}
Composition is defined element-wise: $(f_1, g_1) \fcmp (f_2, g_2) = (f_1 \fcmp f_2,\ g_1 \fcmp g_2)$, and identities are given by pairs of identities from the original categories: $id_{(C, D)} = (id_C, id_D)$.
\end{definition}

Product categories allow us to define bifunctors.

\begin{definition}[Bifunctor]
    A \textit{bifunctor} is a functor whose domain is a product category.
\end{definition}

\begin{definition}[Monoidal Category]
    A \textit{monoidal structure} on a category $\mathcal{C}$ consists of the following:
    \begin{enumerate}[label=(\roman*)]
        \item A bifunctor $- \otimes -: \mathcal{C} \times \mathcal{C} \rightarrow \mathcal{C}$ called the \textit{monoidal product}.
        \item An object $I \in \mathcal{C}$, called the \textit{monoidal unit}.
        \item Three natural isomorphisms:
        \begin{enumerate}
            \item $\alpha: ((- \otimes -) \otimes -)\longrightarrow (- \otimes (- \otimes -))$ called the \textit{associator}, with components $\alpha_{A, B, C}: (A \otimes B) \otimes C \rightarrow A \otimes (B \otimes C)$.
            \item $\lambda: (I \otimes -) \rightarrow (-)$ called the \textit{left unitor}, with components $\lambda_C: I \otimes C \rightarrow C$.
            \item $\rho: (C \otimes I) \rightarrow (-)$ called the \textit{right unitor}, with components $\rho_C: C \otimes I \rightarrow C$.
        \end{enumerate}
    \end{enumerate}
    These constituents are subject to the following \textit{coherence conditions} for all $A, B, C, D \in \mathcal{C}$:
    \begin{enumerate}[label=(\roman*)]
        \item The \textit{pentagon identity}
        \begin{equation*}
            \begin{tikzcd}[column sep=-2em]
                & & (A \otimes B) \otimes (C \otimes D) \arrow[drr, "\alpha_{A, B, C \otimes D}"] & & \\ ((A \otimes B) \otimes C) \otimes D
                 \arrow[urr, "\alpha_{A \otimes B, C, D}"] \arrow[dr, "\alpha_{A, B, C} \otimes id_D", swap] & & & & A \otimes (B \otimes (C \otimes D)) \arrow[dl, "id_A \otimes \alpha_{B, C, D}"] \\[2em] & (A \otimes (B \otimes C)) \otimes D \arrow[rr, "\alpha_{A, B \otimes C, D}", swap]
                & & A \otimes ((B \otimes C) \otimes D) &
            \end{tikzcd}
        \end{equation*}
        \item The \textit{triangle identity}
                \begin{equation*}
            \begin{tikzcd}
                (A \otimes I) \otimes B \arrow[rr, "\alpha_{A, I, B}"] \arrow[dr, "\lambda_A \otimes id_B", swap] & & A \otimes (I \otimes B) \arrow[dl, "id_A \otimes \rho_{B}"] \\
                & A \otimes B &
            \end{tikzcd}
        \end{equation*}
    \end{enumerate}
    We call a category with a particular monoidal structure a \textit{monoidal category}. When we want to indicate the particular monoidal structure of a monoidal category, we will write $(\mathcal{C}, \ \otimes_{\mathcal{C}},\ I_{\mathcal{C}})$.
\end{definition}

\begin{definition}[Symmetric monoidal category]
    A \textit{symmetric monoidal category} is a monoidal category with an additional natural isomorphism $\sigma: (- \otimes -) \rightarrow (- \otimes -)$, called the \textit{swap map}, with components $\sigma_{A, B}: A \otimes B \rightarrow B \otimes B$ such that:
    \begin{enumerate}[label=(\roman*)]
        \item For all $A, B \in \mathcal{C}$, $\sigma_{A, B} \fcmp \sigma_{B, A} = id_{A \otimes B}$. In other words, swapping is self-inverse.
        \item For all $A, B, C \in \mathcal{C}$, the following diagram commutes:
        \begin{equation*}
            \begin{tikzcd}
                (A \otimes B) \otimes C \arrow[r, "\alpha_{A, B, C}"] \arrow[d, "\sigma_{A, B} \otimes id_C"]
                & A \otimes (B \otimes C) \arrow[r, "\sigma_{A, B \otimes C}"]
                & (B \otimes C) \otimes A \arrow[d, "\alpha_{B, C, A}"] \\
                (B \otimes A) \otimes C \arrow[r, "\alpha_{B, A, C}"]
                & B \otimes (A \otimes C) \arrow[r, "id_B \otimes \sigma_{A, C}"]
                & B \otimes (C \otimes A)
            \end{tikzcd}
        \end{equation*}
    \end{enumerate}
\end{definition}

\begin{definition}[Strict monoidal category]
    When the associator and unitors of a monoidal category are all identity morphisms, the coherence conditions hold automatically. We call a monoidal category for which this is the case a \textit{strict monoidal category}.
\end{definition}

\begin{definition}[PROP]
A \textit{PROP} is a strict symmetric monoidal category with objects the natural numbers and monoidal product on objects given by addition.
\end{definition}

\subsection{String Diagrams, Functorial Boxes and Lax Monoidal Functors}

\subsubsection{String Diagrams}

String diagrams are an intuitive, yet formal graphical language with which to express taking arbitray compositions and tensor products of morphisms in strict monoidal categories. In these diagrams, morphisms are drawn as boxes whose inputs and outputs are connected by wires labelled with objects of the category. The way that the boxes are wired together determines a term created by composing and taking the monoidal product of morphisms in the category. Identity morphisms are drawn as a wire, swap maps are drawn as crossing wires, and the monoidal unit is drawn as an empty diagram. For example, for $f : A \rightarrow B$ and $g : B \rightarrow C$ in a symmetric monoidal category, the term:
\begin{equation*}
    (f \otimes id_D)\ \fcmp\ \sigma_{B, D}\ \fcmp\ (id_D \otimes g)
\end{equation*}

becomes:

\begin{equation*}
    \begin{tikzpicture}
	\begin{pgfonlayer}{nodelayer}
		\node [style=medium box] (0) at (-2.5, 1.25) {$f$};
		\node [style=medium box] (1) at (2.5, -1.25) {$g$};
		\node [style=none] (2) at (-4.5, -1.25) {};
		\node [style=none] (3) at (-4.5, 1.25) {};
		\node [style=none] (4) at (4.5, -1.25) {};
		\node [style=none] (5) at (4.5, 1.25) {};
		\node [style=none] (6) at (1.5, -1.25) {};
		\node [style=none] (7) at (1.5, 1.25) {};
		\node [style=none] (8) at (-1.5, 1.25) {};
		\node [style=none] (9) at (-1.5, -1.25) {};
		\node [style=none] (10) at (-4, 1.75) {$A$};
		\node [style=none] (11) at (-1.25, 1.75) {$B$};
		\node [style=none] (12) at (4, -0.75) {$C$};
		\node [style=none] (13) at (4, 2) {$D$};
		\node [style=none] (14) at (1.25, -0.75) {$B$};
	\end{pgfonlayer}
	\begin{pgfonlayer}{edgelayer}
		\draw (3.center) to (0);
		\draw [in=180, out=0] (8.center) to (6.center);
		\draw [in=180, out=0] (9.center) to (7.center);
		\draw (8.center) to (0);
		\draw (2.center) to (9.center);
		\draw (6.center) to (1);
		\draw (7.center) to (5.center);
		\draw (1) to (4.center);
	\end{pgfonlayer}
\end{tikzpicture}
\end{equation*}

String diagrams represent the same morphism up to topological deformation. This gives rise to a key feature of string diagrams: that they automatically capture the associativity and identity laws for both composition and monoidal product. For composition, associativity becomes:

\begin{equation*}
    f \fcmp (g \fcmp h) \quad = \quad \begin{tikzpicture}
	\begin{pgfonlayer}{nodelayer}
		\node [style=medium box] (0) at (-3, 0) {$f$};
		\node [style=medium box] (1) at (0, 0) {$g$};
		\node [style=none] (3) at (-5, 0) {};
		\node [style=medium box] (4) at (3, 0) {$h$};
		\node [style=none] (5) at (5, 0) {};
		\node [style=none] (6) at (-4.5, 0.75) {$A$};
		\node [style=none] (7) at (-1.5, 0.75) {$B$};
		\node [style=none] (8) at (1.5, 0.75) {$C$};
		\node [style=none] (9) at (4.5, 0.75) {$D$};
	\end{pgfonlayer}
	\begin{pgfonlayer}{edgelayer}
		\draw (3.center) to (0);
		\draw (0) to (1);
		\draw (1) to (4);
		\draw (4) to (5.center);
	\end{pgfonlayer}
\end{tikzpicture} \quad = \quad (f \fcmp g) \fcmp h
\end{equation*}

While the identity law $f \fcmp id_B = f = id_A \fcmp f$ becomes:

\begin{equation*}
    \begin{tikzpicture}
	\begin{pgfonlayer}{nodelayer}
		\node [style=medium box] (0) at (-2, 0) {$f$};
		\node [style=none] (3) at (-4, 0) {};
		\node [style=none] (5) at (4, 0) {};
		\node [style=none] (6) at (-3.5, 0.75) {$A$};
		\node [style=none] (7) at (3.5, 0.75) {$B$};
	\end{pgfonlayer}
	\begin{pgfonlayer}{edgelayer}
		\draw (3.center) to (0);
		\draw (0) to (5.center);
	\end{pgfonlayer}
\end{tikzpicture} \quad = \quad \begin{tikzpicture}
	\begin{pgfonlayer}{nodelayer}
		\node [style=medium box] (0) at (0, 0) {$f$};
		\node [style=none] (3) at (-4, 0) {};
		\node [style=none] (5) at (4, 0) {};
		\node [style=none] (6) at (-3.5, 0.75) {$A$};
		\node [style=none] (7) at (3.5, 0.75) {$B$};
	\end{pgfonlayer}
	\begin{pgfonlayer}{edgelayer}
		\draw (3.center) to (0);
		\draw (0) to (5.center);
	\end{pgfonlayer}
\end{tikzpicture} \quad = \quad \begin{tikzpicture}
	\begin{pgfonlayer}{nodelayer}
		\node [style=medium box] (0) at (2, 0) {$f$};
		\node [style=none] (3) at (-4, 0) {};
		\node [style=none] (5) at (4, 0) {};
		\node [style=none] (6) at (-3.5, 0.75) {$A$};
		\node [style=none] (7) at (3.5, 0.75) {$B$};
	\end{pgfonlayer}
	\begin{pgfonlayer}{edgelayer}
		\draw (3.center) to (0);
		\draw (0) to (5.center);
	\end{pgfonlayer}
\end{tikzpicture}
\end{equation*}

And similar for the monoidal product, where the diagram composition now occurs in parallel. For further introductions to string diagrams, see \cite{selinger2010survey, piedeleu2023introduction}.

\subsubsection{Functorial Boxes}

First introduced in \cite{cockett1999linearly} and further developed in \cite{mellies2006functorial}, functorial boxes allow the incorporation of functor application into string diagrams. For a string diagram drawn in some category $\mathcal{D}$, we can depict a functor $\mathsf{F}: \mathcal{C} \longrightarrow \mathcal{D}$ applied to a morphism in $\mathcal{C}$ by containing a string diagram for the morphism in $\mathcal{C}$ within a box:

\begin{equation*}
    \begin{tikzpicture}
	\begin{pgfonlayer}{nodelayer}
		\node [style=medium box] (0) at (0, 0) {$\mathsf{F}(f)$};
		\node [style=none] (1) at (-3.25, 0) {};
		\node [style=none] (2) at (3.25, 0) {};
		\node [style=none] (3) at (-2.5, 0.5) {$\mathsf{F}(A)$};
		\node [style=none] (4) at (2.5, 0.5) {$\mathsf{F}(B)$};
	\end{pgfonlayer}
	\begin{pgfonlayer}{edgelayer}
		\draw (1.center) to (0);
		\draw (0) to (2.center);
	\end{pgfonlayer}
\end{tikzpicture} \quad = \quad \begin{tikzpicture}
	\begin{pgfonlayer}{nodelayer}
		\node [style=medium box] (0) at (0, 0) {$f$};
		\node [style=none] (1) at (-7.25, 0) {};
		\node [style=none] (2) at (7.25, 0) {};
		\node [style=none] (3) at (-6.25, 0.75) {$\mathsf{F}(A)$};
		\node [style=none] (4) at (6.25, 0.75) {$\mathsf{F}(B)$};
		\node [style=none] (5) at (-5, 2) {};
		\node [style=none] (6) at (-5, -2) {};
		\node [style=none] (7) at (5, -2) {};
		\node [style=none] (8) at (5, 2) {};
		\node [style=none] (9) at (-3, 0.75) {$A$};
		\node [style=none] (10) at (3, 0.75) {$B$};
	\end{pgfonlayer}
	\begin{pgfonlayer}{edgelayer}
		\draw (1.center) to (0);
		\draw (0) to (2.center);
		\draw [style=functorbox] (6.center)
			 to (7.center)
			 to (8.center)
			 to (5.center)
			 to cycle;
	\end{pgfonlayer}
\end{tikzpicture}
\end{equation*}

Note that a wire of type $A \in \mathcal{C}$ becomes a wire of type $\mathsf{F}(A) \in \mathcal{D}$ upon exiting the box. These boxes elegantly capture the functor law $\mathrm{F}(id_A) = id_{\mathrm{F}(A)}$ by allowing us to freely add or remove boxes only containing a wire:

\begin{equation*}
    \begin{tikzpicture}
	\begin{pgfonlayer}{nodelayer}
		\node [style=none] (0) at (0, 0) {};
		\node [style=none] (1) at (-6, 0) {};
		\node [style=none] (2) at (6, 0) {};
		\node [style=none] (3) at (-5, 0.75) {$\mathsf{F}(A)$};
		\node [style=none] (4) at (5, 0.75) {$\mathsf{F}(A)$};
	\end{pgfonlayer}
	\begin{pgfonlayer}{edgelayer}
		\draw (1.center) to (0.center);
		\draw (0.center) to (2.center);
	\end{pgfonlayer}
\end{tikzpicture} \quad = \quad \begin{tikzpicture}
	\begin{pgfonlayer}{nodelayer}
		\node [style=none] (0) at (0, 0) {};
		\node [style=none] (1) at (-7.25, 0) {};
		\node [style=none] (2) at (7.25, 0) {};
		\node [style=none] (3) at (-6.25, 0.75) {$\mathsf{F}(A)$};
		\node [style=none] (4) at (6.25, 0.75) {$\mathsf{F}(A)$};
		\node [style=none] (5) at (-5, 2) {};
		\node [style=none] (6) at (-5, -2) {};
		\node [style=none] (7) at (5, -2) {};
		\node [style=none] (8) at (5, 2) {};
		\node [style=none] (9) at (-3, 0.75) {$A$};
		\node [style=none] (10) at (3, 0.75) {$A$};
	\end{pgfonlayer}
	\begin{pgfonlayer}{edgelayer}
		\draw (1.center) to (0.center);
		\draw (0.center) to (2.center);
		\draw [style=functorbox] (6.center)
			 to (7.center)
			 to (8.center)
			 to (5.center)
			 to cycle;
	\end{pgfonlayer}
\end{tikzpicture}
\end{equation*}

As well as capturing the law $\mathrm{F}(f \fcmp g) = \mathrm{F}(f) \fcmp \mathrm{F}(g)$ by allowing us to fuse adjacent boxes:

\begin{equation*}
    \begin{tikzpicture}
	\begin{pgfonlayer}{nodelayer}
		\node [style=medium box] (0) at (-2, 0) {$f$};
		\node [style=none] (1) at (-7.25, 0) {};
		\node [style=none] (2) at (7.25, 0) {};
		\node [style=none] (3) at (-6.25, 0.75) {$\mathsf{F}(A)$};
		\node [style=none] (4) at (6.25, 0.75) {$\mathsf{F}(C)$};
		\node [style=none] (5) at (-5, 2) {};
		\node [style=none] (6) at (-5, -2) {};
		\node [style=none] (7) at (5, -2) {};
		\node [style=none] (8) at (5, 2) {};
		\node [style=none] (9) at (-3.75, 0.75) {$A$};
		\node [style=none] (10) at (0, 0.75) {$B$};
		\node [style=medium box] (11) at (2, 0) {$g$};
		\node [style=none] (12) at (3.75, 0.75) {$C$};
	\end{pgfonlayer}
	\begin{pgfonlayer}{edgelayer}
		\draw (1.center) to (0);
		\draw [style=functorbox] (6.center)
			 to (7.center)
			 to (8.center)
			 to (5.center)
			 to cycle;
		\draw (0) to (11);
		\draw (11) to (2.center);
	\end{pgfonlayer}
\end{tikzpicture} \quad = \quad \begin{tikzpicture}
	\begin{pgfonlayer}{nodelayer}
		\node [style=medium box] (0) at (-3.25, 0) {$f$};
		\node [style=none] (1) at (-7.75, 0) {};
		\node [style=none] (2) at (7.75, 0) {};
		\node [style=none] (3) at (-6.75, 0.75) {$\mathsf{F}(A)$};
		\node [style=none] (4) at (6.75, 0.75) {$\mathsf{F}(C)$};
		\node [style=none] (5) at (-5.25, 2) {};
		\node [style=none] (6) at (-5.25, -2) {};
		\node [style=none] (7) at (-1.25, -2) {};
		\node [style=none] (8) at (-1.25, 2) {};
		\node [style=none] (9) at (-4.5, 0.75) {$A$};
		\node [style=none] (10) at (0, 0.75) {$\mathsf{F}(B)$};
		\node [style=medium box] (11) at (3.25, 0) {$g$};
		\node [style=none] (12) at (4.5, 0.75) {$C$};
		\node [style=none] (13) at (1.25, 2) {};
		\node [style=none] (14) at (1.25, -2) {};
		\node [style=none] (15) at (5.25, -2) {};
		\node [style=none] (16) at (5.25, 2) {};
		\node [style=none] (17) at (-2, 0.75) {$B$};
		\node [style=none] (18) at (2, 0.75) {$B$};
	\end{pgfonlayer}
	\begin{pgfonlayer}{edgelayer}
		\draw (1.center) to (0);
		\draw [style=functorbox] (6.center)
			 to (7.center)
			 to (8.center)
			 to (5.center)
			 to cycle;
		\draw (0) to (11);
		\draw (11) to (2.center);
		\draw [style=functorbox] (14.center)
			 to (15.center)
			 to (16.center)
			 to (13.center)
			 to cycle;
	\end{pgfonlayer}
\end{tikzpicture}
\end{equation*}

\subsubsection{Lax Monoidal Functors}

There are subtleties when applying functors to monoidal categories in the way that a functor interacts with the monoidal product. This becomes important for our functorial boxes. For example, $\mathsf{F}(A \otimes B)$ may not necessarily equal $\mathsf{F}(A) \otimes \mathsf{F}(B)$. This gives rise to the notion of lax, oplax and strong monoidal functors. In terms of string diagrams and functorial boxes, this governs the ability to allow whether multiple wires can enter a box, and whether multiple wires exiting a box are preserved as multiple wires or combined into a single wire. In our work, we are interested in lax monoidal functors\cite{nlabmonoidalfunctor}:

\begin{definition}[Lax Monoidal Functor]
    Let $(\mathcal{C}, \ \otimes_{\mathcal{C}},\ I_{\mathcal{C}})$ and $(\mathcal{D}, \ \otimes_{\mathcal{D}},\ I_{\mathcal{D}})$ be monoidal categories. A \textit{lax monoidal functor} between them consists of:
    \begin{enumerate}[label=(\roman*)]
        \item A functor $F: \mathcal{C} \rightarrow \mathcal{D}$.
        \item \textit{Coherence maps}:
        \begin{enumerate}
            \item A morphism $\epsilon: I_{\mathcal{D}} \rightarrow F(I_{\mathcal{C}})$ in $\mathcal{D}$.
            \item A natural transformation $\mu: F(-) \otimes_{\mathcal{D}} F(-) \longrightarrow F(- \otimes_{\mathcal{C}} -)$, with components $\mu_{A, B}: F(A) \otimes_{\mathcal{D}} F(B) \rightarrow F(A \otimes_{\mathcal{C}} B)$.
        \end{enumerate}
    \end{enumerate}
    These coherence maps are subject to the following conditions:
    \begin{enumerate}[label=(\roman*)]
        \item \textit{Associativity}: for all $A, B, C \in \mathcal{C}$, the following diagram commutes:
        \begin{equation*}
            \begin{tikzcd}[column sep=4em, row sep=3.5em]
                (F(A) \otimes_{\mathcal{D}} F(B)) \otimes_{\mathcal{D}} F(C) \arrow[r, "\alpha^{\mathcal{D}}_{F(A), F(B), F(C)}"] \arrow[d, "\mu_{A, B} \otimes_{\mathcal{D}} id_{F(C)}", swap]
                & F(A) \otimes_{\mathcal{D}} (F(B) \otimes_{\mathcal{D}} F(C)) \arrow[d, "id_{F(A)} \otimes_{\mathcal{D}} \mu_{B, C}"]
                \\ F(A \otimes_{\mathcal{C}} B) \otimes_{\mathcal{D}} F(C) \arrow[d, "\mu_{A \otimes_{\mathcal{C}} B, C}"]
                & F(A) \otimes_{\mathcal{D}} F(B \otimes_{\mathcal{C}} C) \arrow[d, "\mu_{A, B \otimes_{\mathcal{C}} C}"] \\
                F((A \otimes_{\mathcal{C}} B) \otimes_{\mathcal{C}} C) \arrow[r, "F
                (\alpha^{\mathcal{C}}_{A, B, C})"] &
                F(A \otimes_{\mathcal{C}} (B \otimes_{\mathcal{C}} C))
            \end{tikzcd}
        \end{equation*}
        Where $\alpha^{\mathcal{C}}$ and $\alpha^{\mathcal{D}}$ denote the associators in $\mathcal{C}$ and $\mathcal{D}$, respectively.
        \item \textit{Unitality}: for all $C \in \mathcal{C}$, the following diagrams commute:
        \begin{equation*}
            \begin{tikzcd}[column sep=4em, row sep=3em]
                I_{\mathcal{D}} \otimes_{\mathcal{D}} F(C) \arrow[r, "\epsilon \otimes_{\mathcal{D}} F(id_C)"]
                \arrow[d, "\lambda^{\mathcal{D}}_{F(C)}", swap]
                & F(I_{\mathcal{C}}) \otimes_{\mathcal{D}} F(C) \arrow[d, "\mu_{I_{\mathcal{C}}, C}"] \\
                F(C) & F(I_{\mathcal{C}} \otimes_{\mathcal{C}} C) \arrow[l, "F(\lambda^{\mathcal{C}}_C)"]
            \end{tikzcd}
        \end{equation*}
        \begin{equation*}
            \begin{tikzcd}[column sep=4em, row sep=3em]
                F(C) \otimes_{\mathcal{D}} I_{\mathcal{D}} \arrow[r, "F(id_C) \otimes_{\mathcal{D}} \epsilon"]
                \arrow[d, "\rho^{\mathcal{D}}_{F(C)}", swap]
                & F(C) \otimes_{\mathcal{D}} F(I_{\mathcal{C}}) \arrow[d, "\mu_{C, I_{\mathcal{C}}}"] \\
                F(C) & F(C \otimes_{\mathcal{C}} I_{\mathcal{C}}) \arrow[l, "F(\rho^{\mathcal{C}}_C)"]
            \end{tikzcd}
        \end{equation*}
        Where $\lambda^{\mathcal{C}},\ \rho^{\mathcal{C}}$ and $\lambda^{\mathcal{D}},\ \rho^{\mathcal{D}}$ are the unitors in $\mathcal{C}$ and $\mathcal{D}$, respectively.
    \end{enumerate}
\end{definition}

The significance of being lax monoidal in the context of functorial boxes is that it allows for multiple wires to entering a box, which are preserved as multiple wires. This is illustrated by the following:

\begin{equation*}
    \begin{tikzpicture}
	\begin{pgfonlayer}{nodelayer}
		\node [style=none] (0) at (0, 2.75) {};
		\node [style=none] (1) at (0, -2.75) {};
		\node [style=none] (2) at (0, 0.6) {};
		\node [style=none] (3) at (-3, 0.6) {};
		\node [style=none] (5) at (0, 1.9) {};
		\node [style=none] (6) at (0, -2.1) {};
		\node [style=none] (7) at (-3, 1.9) {};
		\node [style=none] (8) at (-3, -2.1) {};
		\node [style=none] (9) at (0, -0.8) {};
		\node [style=none] (10) at (-3, -0.8) {};
		\node [style=none] (11) at (3, 1.9) {};
		\node [style=none] (12) at (3, 0.6) {};
		\node [style=none] (13) at (3, -0.8) {};
		\node [style=none] (14) at (3, -2.1) {};
		\node [style=none] (15) at (3.25, 2.75) {};
		\node [style=none] (16) at (3.25, -2.75) {};
		\node [style=none] (17) at (-2.5, 2.35) {\footnotesize$\mathsf{F}(A)$};
		\node [style=none] (18) at (-2.5, -0.35) {\footnotesize$\mathsf{F}(C)$};
		\node [style=none] (19) at (-2.5, 1.05) {\footnotesize$\mathsf{F}(B)$};
		\node [style=none] (20) at (-2.5, -1.7) {\footnotesize$\mathsf{F}(D)$};
		\node [style=none] (21) at (2, 2.275) {\footnotesize$A$};
		\node [style=none] (22) at (2, 1) {\footnotesize$B$};
		\node [style=none] (23) at (2, -0.4) {\footnotesize$C$};
		\node [style=none] (24) at (2, -1.75) {\footnotesize$D$};
	\end{pgfonlayer}
	\begin{pgfonlayer}{edgelayer}
		\draw [style=functorbox] (16.center)
			 to (1.center)
			 to (0.center)
			 to (15.center);
		\draw (7.center) to (5.center);
		\draw (3.center) to (2.center);
		\draw (8.center) to (6.center);
		\draw (10.center) to (9.center);
		\draw (2.center) to (12.center);
		\draw (5.center) to (11.center);
		\draw (9.center) to (13.center);
		\draw (6.center) to (14.center);
	\end{pgfonlayer}
\end{tikzpicture}
\end{equation*}

Without the lax monoidal condition, only one wire could enter each  box. As there would not be a sensible notion of multiple wires entering, any diagram containing such would be considered ill-formed. 

An oplax monoidal functor is defined similarly to a lax monoidal functor, except the coherence maps go in the other direction:
\begin{equation*}
    \eta: F(I_{\mathcal{C}}) \rightarrow I_{\mathcal{D}} \quad \quad \quad \quad \nu: F(- \otimes_{\mathcal{C}} -) \longrightarrow F(-)
\end{equation*}

An oplax monoidal functor allows multiple wires exiting a box to be preserved as multiple wires. If this is not the case wires exiting are bundled together into a single wire as follows:

\begin{equation*}
    \begin{tikzpicture}
	\begin{pgfonlayer}{nodelayer}
		\node [style=none] (0) at (0, 2.75) {};
		\node [style=none] (1) at (0, -2.75) {};
		\node [style=none] (2) at (0, 0.5) {};
		\node [style=none] (3) at (-3, 0.5) {};
		\node [style=none] (4) at (3.25, -0.1) {};
		\node [style=none] (5) at (0, 1.8) {};
		\node [style=none] (6) at (0, -2.2) {};
		\node [style=none] (7) at (-3, 1.8) {};
		\node [style=none] (8) at (-3, -2.2) {};
		\node [style=none] (9) at (0, -0.9) {};
		\node [style=none] (10) at (-3, -0.9) {};
		\node [style=none] (11) at (0, -0.1) {};
		\node [style=none] (12) at (-3.25, 2.75) {};
		\node [style=none] (13) at (-3.25, -2.75) {};
		\node [style=none] (14) at (-2.5, 2.2) {\footnotesize$A$};
		\node [style=none] (15) at (-2.5, 0.9) {\footnotesize$B$};
		\node [style=none] (16) at (-2.5, -0.5) {\footnotesize$C$};
		\node [style=none] (17) at (-2.5, -1.8) {\footnotesize$D$};
		\node [style=none] (18) at (3.25, 0.4) {{\footnotesize$\mathsf{F}(A \otimes B \otimes C \otimes D)$}};
	\end{pgfonlayer}
	\begin{pgfonlayer}{edgelayer}
		\draw [style=functorbox] (13.center)
			 to (1.center)
			 to (0.center)
			 to (12.center);
		\draw (7.center) to (5.center);
		\draw (3.center) to (2.center);
		\draw (8.center) to (6.center);
		\draw (10.center) to (9.center);
		\draw (11.center) to (4.center);
	\end{pgfonlayer}
\end{tikzpicture}
\end{equation*}

If the coherence maps are all isomorphisms for either a lax or an oplax monoidal functor, it is called a strong monoidal functor, and multiple wires can both enter and exit the functorial boxes.

\end{document}